\documentclass[twocolumn,journal]{IEEEtran}
\ifCLASSINFOpdf
\else
\fi

\usepackage{url,comment}
\usepackage{ifthen}
\usepackage{cite,color}
\usepackage[cmex10]{amsmath} 
                             
\newcommand{\stirlingii}{\genfrac{\{}{\}}{0pt}{}}
\usepackage[final]{graphicx}

\usepackage{amsthm}
\usepackage{amsthm}
\usepackage{mathtools}
\usepackage{amssymb}
\usepackage{euscript}
\usepackage{subfigure}
\usepackage[colorinlistoftodos]{todonotes}
\usepackage[a-1b]{pdfx}
\usepackage{hyperref}
\usepackage{bbm}
\usepackage{flushend}
\usepackage{enumerate}

\newtheorem{theorem}{Theorem}
\newtheorem{definition}{Definition}
\newtheorem{lemma}{Lemma}

\newtheorem{corollary}{Corollary}

\begin{document}
\title{Efficient Replication for Straggler Mitigation in Distributed Computing}
\author{\IEEEauthorblockN{Amir Behrouzi-Far and Emina Soljanin}\\
\IEEEauthorblockA{Department of Electrical and Computer Engineering, Rutgers University} \\
\{amir.behrouzifar,emina.soljanin\}@rutgers.edu}

\maketitle

\begin{abstract}



Master-worker distributed computing systems
use task replication in order to mitigate the effect of slow workers, known as stragglers. Tasks are grouped into batches and assigned to one or more workers for execution. We first consider the case when the batches do not overlap and, 
using the results from majorization theory, show that, for a general class of workers' service time distributions, a balanced assignment of batches to workers minimizes the average job compute time. We next show that this balanced assignment of non-overlapping batches achieves lower average job compute time compared to the  overlapping schemes proposed in the literature. Furthermore, we derive the optimum redundancy level as a function of the service time distribution at workers. We show that the redundancy level that minimizes average job compute time is not necessarily the same as the redundancy level that maximizes the predictability of job compute time, and thus there exists a trade-off between optimizing the two metrics. Finally, by running experiments on Google cluster traces, we observe that redundancy can reduce the compute time of the jobs in Google clusters by an order of magnitude, and that the optimum level of redundancy depends on the distribution of tasks' service time.
\end{abstract}
\IEEEpeerreviewmaketitle

\section{Introduction}

{\let\thefootnote\relax\footnote{This work was presented in part in Annual Allerton Conference on Communication, Control and Computing \cite{behrouzi2018effect} and IEEE international conference on Big Data \cite{behrouzi2019data}, and supported in part by the NSF award No.~CIF-1717314.
}}

\IEEEPARstart{D}{istributed} computing plays an essential role in modern data analytics and machine learning systems \cite{kraska2013mlbase,buchlovsky2019tf}. Parallel task execution on multiple nodes can bring considerable speedups to many practical applications,  e.g., matrix multiplication \cite{benson2015framework}, model training in machine learning \cite{dean2012large} and convex optimization \cite{boyd2011distributed}.
However, distributed computing systems are prone to node failure and slowdowns. 

In a master-worker architecture, where the master node waits for each worker to deliver its computation results (before possibly moving to the next stage of an algorithm), the latency is determined by the slowest workers, known as \textit{``stragglers''} \cite{dean2013tail}. Straggler mitigation by task replication has been considered in, e.g., \cite{wang2014efficient,gardner2015reducing,joshi2017efficient,joshi2015efficient}, and by erasure coding in, e.g., \cite{aktas2017effective,aktacs2019straggler}. Replication is easier to implement but erasure coding is generally more efficient (in terms of resource usage) \cite{aktacs2019straggler}. Furthermore, several work focus on performance analysis of more systems level techniques \cite{aktas2018straggler,ozfatura2019speeding,ferdinand2018anytime}.

We are concerned with distributed computing systems with master-worker architecture, which is used by several commercial distributed computing engines, e.g., MapReduce and Kubernetes, and also in theoretical papers, e.g., \cite{tandon2016gradient,raviv2017gradient,ferdinand2018anytime}.  With redundancy, some workers are selected as backups and thus the system can tolerate some number of stragglers. Especially, the current literature on coded computing, e.g. \cite{yu2019lagrange}, measures the performance of redundancy techniques by the number of stragglers a system can tolerate. However, the core performance metric in practical systems, such as Google \cite{dean2012large}, Amazon \cite{jackson2010performance} and Facebook \cite{hazelwood2018applied}, is the \textit{latency} that jobs experience; see also \cite{peng2020diversityparallelism} and references therein. Minimizing the average latency and maximizing the predictability of latency are among the main concerns in practical systems \cite{dean2013tail}. Among other results, we show that, redundancy techniques that tolerate the same number of stragglers can have different average latency and, therefore, not equally effective in a practical system.

Most computing jobs in machine learning and data analytics applications require running an algorithm on a big dataset \cite{dean2012large}. These jobs can be effectively distributed among workers, by assigning to each of them a fraction of the dataset and aggregating their results\cite{tandon2016gradient}. For example, the computing jobs that involve multiplication of matrices can be distributed effectively, by assigning rows and columns to the worker nodes and recombining the results by the master node \cite{bulucc2012parallel}. Such computing jobs (parallelizable to smaller tasks) are common and have motivated this research. 

In this paper, we derive the efficient assignment of replicated redundancy, in a master-worker computing model, where jobs are parallelizable to smaller tasks. For a fixed level of redundancy, we find the optimum task assignment to the workers which minimizes the average job compute time. For such assignments, we find the optimum redundancy level that minimizes the average job compute time, and the (often different) optimum redundancy level that maximizes the compute time predictability. Our specific contributions are as outlined below.

First, given a budget of $N$ workers and a job with $N$ parallel tasks and a pre-defined level of redundancy, we derive the optimum task assignment that minimizes the average job compute time. With a pre-defined level of redundancy, each worker is assigned a \textit{batch} of tasks. Using the results from majorization theory, we show that if the Complementary Cumulative Density Function (CCDF) of the batch compute time is stochastically decreasing and convex in the number of workers hosting the batch, then the minimum average job compute time is achieved by balanced assignment of non-overlapping batches (see Fig. \ref{fig:taskDist}). We then argue that, even though the previously proposed overlapping \cite{tandon2016gradient} or non-balanced \cite{li2017near} batch assignments can tolerate (with the same redundancy) the same number of stragglers as the balanced assignment of non-overlapping batches, they are not optimal for minimizing the average job compute time.

Second, with the balanced assignment of non-overlapping batches, we derive the optimum level of redundancy for different service time distributions at workers. We observe that the higher the randomness in workers' service time the higher the optimum level of redundancy. Furthermore, 
when the service time distribution at workers is heavy-tailed, the benefits are larger than when the service times are exponential.

Third, we study compute the \textit{Coefficient of Variations} (CoV) of job compute time to measure the degree of \textit{predictability} of job compute time. CoV is among the most important performance metrics both in theory \cite{harchol2000implementation} and in practice\cite{dean2013tail}. We observe that the level of redundancy that minimizes the coefficient of variations of job compute time is not necessarily the same as the level that minimizes the average job compute time. Thus, there exists a trade-off between the two metrics that can be regulated by the redundancy level. 

Finally, we run experiments on the Google cluster traces, using the runtime information of several jobs in the dataset. We observe that, jobs with both heavy-tail and (shifted) exponential distributions of compute time are present in Google clusters. We confirm our theoretical findings that 1) replication can reduce the average compute time, 2) the optimum level of redundancy depends on the distribution of service time at workers, and 3) the jobs with heavy-tail distribution of service time benefit more from redundancy.

This paper is organized as follows: In Sec.~\ref{sectionII}, we describe the system architecture and compute job model, task assignment model and the probability distributions we use throughout the paper. The task batching schemes, non-overlapping and overlapping batching, are described in the same section. In Section \ref{sectionIV}, we study the optimum assignment of non-overlapping batches. In Sec.~\ref{section V}, we compare the average job compute time of overlapping batches and non-overlapping batches. Optimum redundancy levels that minimizing the average job compute time and maximize the compute time predictability are studied in Sec.~\ref{sectionVI}. We present our experimental results in Section \ref{sectionVII} and conclude in Sec.~\ref{sectionVIII}. All the proofs are presented in the Appendix.

\begin{figure}[t]
   \centering
   \includegraphics[width=7cm, keepaspectratio]{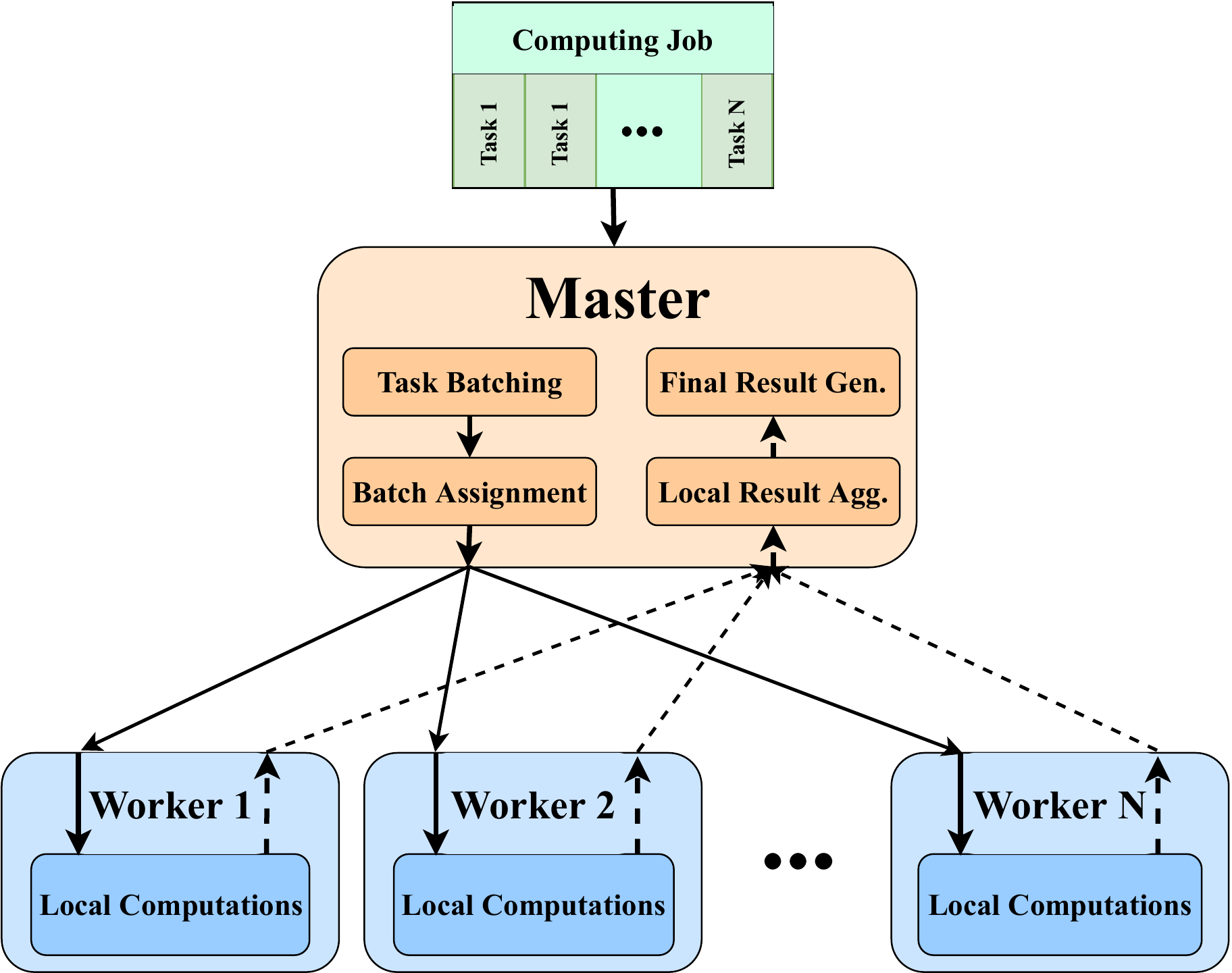}
   \caption{The master-worker architecture considered in this work. The master node (redundantly) assigns the tasks  to the worker nodes. Upon receiving the computation results from a large enough group of workers, the master node generates the overall result.}
   \label{fig:sysModel1}
\end{figure}

\section{System Model}\label{sectionII}

\subsection{System Architecture and Task Batching Schemes}
We study a master-worker architecture, as shown in Fig.~\ref{fig:sysModel1}, and refer to it as system \ref{fig:sysModel1}. Each job is $N$-parallelizable, that is, it consists of $N$ smaller {\it tasks} that can be  concurrently executed in $N$ workers. The master node 1) forms $B$ batches of tasks (task batching), 2) assigns batches to workers (batch assignment), 3) aggregates computing results from the workers (local result aggregation), and 4) generates the overall result (final result generating).  A worker executes all its tasks and only then communicates the aggregate results to the master. Thus, for each batch assigned to it, a worker communicates only once with the master, which makes the task execution order within a batch inconsequential. 

Each of the $N$ task is \textit{redundantly} assigned to  one or more batches and each of the $B$ batch is assigned to at least one worker. There are at least as many workers as batches, that is, $N\ge B$. A batch is completed as soon as one of its replicas is completed.  This model has been used for a wide range of problems, e.g., matrix multiplication \cite{lee2017high}, gradient based optimizers \cite{boyd2011distributed}, model training in machine learning problems \cite{tandon2016gradient}. Batches can have tasks in common (overlapping batches) or not (non-overlapping batches). 

\subsubsection{Non-overlapping Batches}
Here, the $N$ tasks are partitioned into $B$ batches, each of size $N/B$  (see Fig. \ref{fig:taskDist}). The two extreme cases are $B=1$, when all the tasks are in the same batch, and $B=N$, when each batch consists of a single task. With $B<N$, a batch can be assigned to more than one worker. A balanced assignment assigns each batch to a fixed number of workers. The random assignment of non-overlapping batches to workers was considered in the literature \cite{li2017near}. In Section \ref{sectionIV}, we show that the resulting imbalanced distribution of batches among workers is not optimum in terms of average job compute time. Furthermore, random distribution of batches may leave some batches not selected at all. This is shown analytically in Appendix A. In that case, the results of the computations will not be accurate, since they are based on a subset of tasks.

\subsubsection{Overlapping Batches}
Here, the $N$ tasks are redundantly grouped in $N$ batches (see Fig.~\ref{fig:taskDist}). The batch assignment here is not an issue since each batch is assigned to a single worker. But we do need to decide which batch size to use and how the batches should overlap in order to minimizes the expected job compute time. In its most general form, this is a hard combinatorial optimization problem, and we will study only the batching schemes that satisfy certain properties, as follows. The overlapping batches could be grouped into smaller subsets, such that each task appears exactly once in each subset. In Fig. \ref{fig:taskDist}, the batches at $W_1$ and $W_3$ make one subset and the batches at $W_2$ and $W_4$ make another subset. The reason that we consider this particular batching schemes is that, 1) the number of workers assigned to each task is fixed (there is no task preference), and 2) a task can appear only once at each worker and thus each task experiences maximum diversity in its service time. Batching schemes with this property have been proposed in the literature, cf.\ \cite{tandon2016gradient,li2017near}. With this property, the only question is that, within a subset, how should tasks be assigned to the batches. This question will be answered in Section \ref{section V}.

\subsubsection{An Example}
Consider the problem of optimizing a model $\Bar{\beta}$ with distributed gradient descent algorithm. The goal is to minimize a loss function $L(\Bar{\beta};\mathbb{D})$, where $\mathbb{D}=\{D_1,D_2,D_3,D_4\}$ is the dataset. Here, $D_i$s are disjoint subsets of $\mathbb{D}$. At the $i$th iteration of the algorithm, the model is given by
 $
        \Bar{\beta}_i=\Bar{\beta}_{i-1}-\gamma\nabla L(\Bar{\beta}_{i-1},\mathbb{D}),
$
where $\gamma$ is the step size. One can rewrite this equation as
$
        \Bar{\beta}_i=\Bar{\beta}_{i-1}-\gamma\nabla\sum_{k=1}^4L(\Bar{\beta}_{i-1},D_k).
$
Here, the summation can be redundantly distributed among four workers. Two examples of assignments are provided in Fig.~\ref{fig:taskDist}. With the non-overlapping assignment, the results from the fastest worker among $W_1$, $W_3$ and the fastest worker among $W_2$, $W_4$ are enough for the master node to generate the overall result. If $X_i$ is the service time random variable (RV) at worker $i$, with the non-overlapping assignment, the job compute time is $\max\{\min\{X_1,X_3\}\min\{X_2,X_4\}\}$. On the other hand, with the overlapping assignment, the result form the fastest group of workers among $W_1$,$W_3$ and $W_2$,$W_4$ is sufficient to the overall result. Thus, the job compute time is $\min\{\max\{X_1,X_3\}\max\{X_2,X_4\}\}$. Accordingly, the statistics of job compute time varies with the type of assignment.

\begin{figure}[t]
   \centering
   \includegraphics[width=\columnwidth, keepaspectratio]{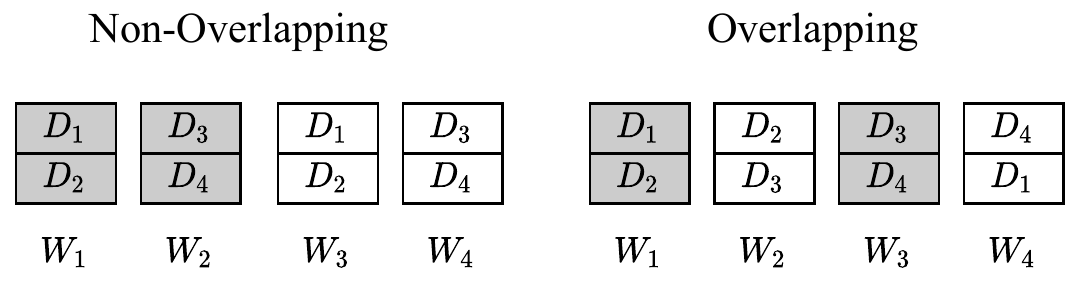}
   \caption{Task replication policies. Either the dark group or the white group are enough for generating the overall compute result.}
   \label{fig:taskDist}
\end{figure}

\subsection{Service Time Model}
We assume the tasks have equal \textit{size} corresponding to their minimum execution time. This is a common assumption in coded computing literature \cite{tandon2016gradient}. In particular, in a master-worker architecture, the master node can split a job into tasks of equal size and assign them to the worker nodes. Workers take random times to execute tasks, and they are statistically identical. The \textit{job compute time} is the time that it takes to complete a job. For a given compute job and a fixed number of statistically identical worker nodes, the statistics of job compute time is determined by 1) the task-to-worker assignment policy and 2) the service time distribution at workers. The following common service time distributions are considered.

\begin{enumerate}[i.]
\item Exponential distribution with rate $\mu$, $X \sim \textup{Exp}(\mu)$ where
    \begin{equation}
        \textup{Pr}\{X>x\}=\mathbbm{1}(x\geq 0)\textup{e}^{-\mu x}.
    \end{equation}
\item Shifted-exponential distribution with shift $\Delta$ and with rate $\mu$,
$X \sim \textup{SExp}(\mu,\Delta)$ where
    \begin{equation}
        \textup{Pr}\{X>x\}=1-\mathbbm{1}(x\geq \Delta)\bigl[1-\textup{e}^{-\mu(x-\Delta)}\bigr],
    \end{equation}
\item Pareto distribution with shape $\alpha$ and scale $\sigma$,\\ 
$X \sim \textup{Pareto}(\alpha,\sigma)$ where
    \begin{equation}
        \textup{Pr}\{X>x\}=1-\mathbbm{1}(x\geq\sigma)\Bigl[1-\left(\frac{x}{\sigma}\right)^{-\alpha}\Bigr].
    \end{equation}
\end{enumerate}

\section{Compute Time with Non-overlapping Batches}\label{sectionIV}
In this section, we show that if the batch compute time is stochastically decreasing and convex random variable then the balanced assignment of non-overlapping batches achieves the minimum average job compute time. Let $N_i$ be the number of workers hosting batch $i$, and define
$\EuScript{\Bar{N}}=(N_1,\dots,N_B)$ as the \textit{batch assignment vector} (cf. Fig.~\ref{fig:fork}). The computations of batch $i$ is completed as soon as one of its $N_i$ host workers finishes the computations, which we refer to as \textit{batch compute time}. Thus, the compute time of the $i$th batch is the minimum of $N_i$ i.i.d RVs:
    \begin{equation}
        T_i=\min\left(T_{i,1},T_{i,2},\dots,T_{i,N_i}\right), \hspace{0.2cm} \forall i\in \{1,2,\dots,B\}.
    \label{min}
    \end{equation}
    
    \begin{figure}[t]
        \centering
            \includegraphics[width=6cm, height=5cm, trim={14 0 0 0},clip]{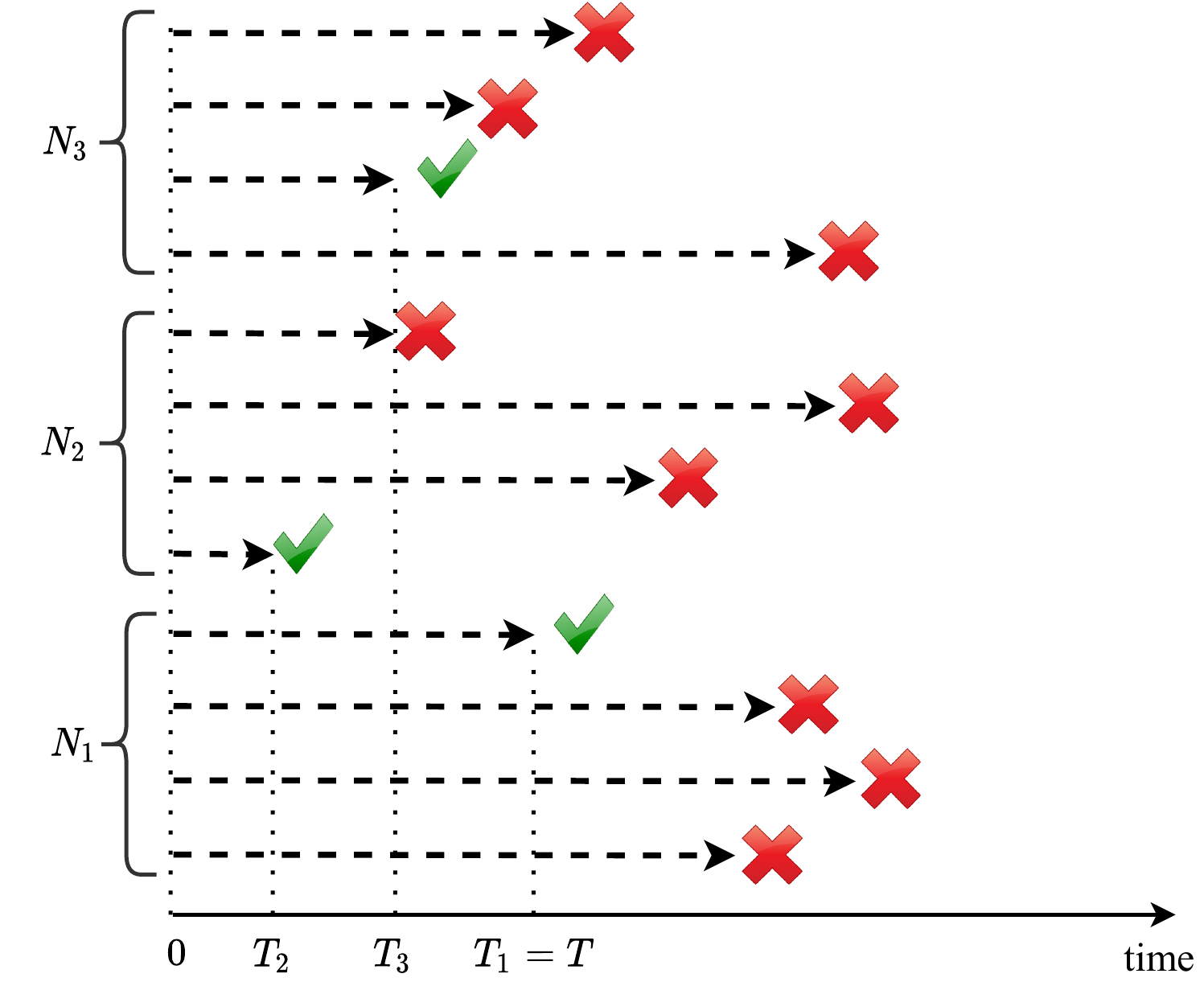}
            \caption{Time diagram of system \ref{fig:sysModel1}. Batch $i$ is assigned to $N_i=4$ workers. The computations of batch $i$ is completed as soon as one of its 4 host workers finishes the computations. Since the master node requires the results of all batches, it has to wait until $T_1$, at which the slowest group $N_1$ of workers, assigned to batch 1, complete the computations.}
        \label{fig:fork}
    \end{figure}
For generating the overall result, the master node has to wait for the results of computations over all batches. In other words, the overall result can be generated only after the slowest group of workers (hosting the same batch) deliver the local result. Hence, the job compute time $T$ could be written as,
\begin{equation}
    T=\max\left(T_1,T_2,\dots,T_B\right).
\label{max}
\end{equation}
To study the effect of batch assignment on the average job compute time $T$, we need the following definitions.

\begin{definition}
The real valued RV $X$ is greater than or equal to the real valued RV $Y$ in the sense of \textbf{usual stochastic ordering}, shown by $X\underset{st}{\geq} Y$, if their CCDF satisfy
    \begin{equation}
        \textup{Pr}\{X>\beta\} \geq \textup{Pr}\{Y>\beta\}, \quad \forall\beta\in \CMcal{R},
    \end{equation}
i.e.,
    $
        \mathbbm{E}[\phi(X)] \geq \mathbbm{E}[\phi(Y)],
    $
for any non-decreasing function $\phi$.
\label{def1}
\end{definition}

\begin{definition}
The RV $X(\theta)$ is \textbf{stochastically decreasing and convex} in $\theta$ if its CCDF is a decreasing and convex in $\theta$.
\label{decreasing}
\end{definition}

\begin{definition}
For any $V_p=(v_{p1},v_{p2},\dots,v_{pM})$ in $\mathbbm{R}^M$, the \textbf{rearranged coordinate vector} $V_{[p]}$ is defined as $V_{[p]}=(v_{[p1]},v_{[p2]},\dots,v_{[pM]})$, the elements of which are the elements of $V$ rearranged in decreasing order, i.e, $v_{[p1]}>\dots>v_{[pM]}$.
\label{rearranged}
\end{definition}

\begin{definition}
Let $V_{[p]}=\left(v_{[p1]},v_{[p2]},\dots,v_{[pM]}\right)$ and $V_{[q]}=\left(v_{[q1]},v_{[q2]},\dots,v_{[qM]}\right)$  be two rearranged coordinate vectors in $\mathbbm{R}^M$. Then $V_p$ \textbf{majorizes} $V_q$, denoted by $V_p\succeq V_q$, if
    \begin{equation*}
        \begin{split}
            &\sum_{i=1}^{m}v_{[pi]}\geq \sum_{i=1}^{m}v_{[qi]}, \enskip \forall m\in \{1,2,\dots,M\},\\
            &\sum_{i=1}^{M}v_{[pi]}= \sum_{i=1}^{M}v_{[qi]}.
        \end{split}
   \end{equation*}
\end{definition}

\begin{definition}
Function $\phi:\mathbbm{R}^M\rightarrow \mathbbm{R}$ is \textbf{schur convex} if for every $V$ and $W$ in $\mathbbm{R}^M$, $V\succeq W$ implies $\phi(V)\geq \phi(W)$.
\end{definition}

\begin{definition}
Real valued RV $Z(\Bar{x})$, $\Bar{x}\in \mathbbm{R}^M$, is \textbf{stochastically schur convex} in $\Bar{x}$, in the sense of usual stochastic ordering, if for any $\Bar{x}$ and $\Bar{y}$ in $\mathbbm{R}^M$, $\Bar{x}\succeq\Bar{y}$ implies $Z(\Bar{x})\underset{st}{\geq} Z(\Bar{y})$.
\label{def6}
\end{definition}

The following lemmas give the batch assignment that minimizes job compute time, under an assumption about the distribution of batch compute time. All the proofs hereafter are postponed to the Appendix.

\begin{lemma}
If the batch assignment $\EuScript{\Bar{N}}_{1}=(N_{11},\dots,N_{1B})$
majorizes the batch assignment $\EuScript{\Bar{N}}_{2}=(N_{21},N_{22},\dots,N_{2B})$, and the batch compute times are i.i.d stochastically decreasing, convex RVs, then the corresponding job compute times $T(\EuScript{\Bar{N}}_{1})$ and $T(\EuScript{\Bar{N}}_{2})$ satisfy
 $
        \mathbbm{E}[ T(\EuScript{\Bar{N}}_{1})] \geq  \mathbbm{E}[ T(\EuScript{\Bar{N}}_{2})],
$
\label{lem}
\end{lemma}

\begin{proof}
The job compute time for the batch assignment vector $\EuScript{\Bar{N}}_{k}$ $\forall k\in\{1,2\}$, is given by
    $
         T(\EuScript{\Bar{N}}_{k})=\max\left(T_{k1},T_{k2},\dots,T_{kB}\right),
    $
where $T_{ki}$ $i\in\{1,2,\dots,B\}$ is stochastically decreasing and convex in $N_{ki}$. Since $\max(\cdot)$ is a schur convex function, $ T(\EuScript{\Bar{N}}_{k})$ is stochastically decreasing and schur convex function of $\EuScript{\Bar{N}}_{k}$ (see \cite{liyanage1992allocation} section 2 for detailed discussion). Hence, by Definition \ref{def6}, $\EuScript{\Bar{N}}_{1}\succeq\EuScript{\Bar{N}}_{2}$ implies $T(\EuScript{\Bar{N}}_{1}) \underset{st}{\geq} T(\EuScript{\Bar{N}}_{2})$ in the sense of usual stochastic ordering, which in turn implies that for any non-decreasing function $\phi$, we have
    $
         \mathbbm{E}[\phi(T(\EuScript{\Bar{N}}_{1}))] \geq \mathbbm{E}[\phi( T(\EuScript{\Bar{N}}_{2}))].
    $
Substituting $\phi$ by the unit ramp function completes the proof.
\end{proof}

\begin{lemma}
The balanced batch assignment, defined as $\EuScript{\Bar{N}}_{b}=(N/B,N/B,\dots,N/B)$ with $N$ and $B$ being the respective number of workers and batches, is majorized by any other batch assignment policy.
\label{majorization}
\end{lemma}

From Lemma \ref{lem} and Lemma \ref{majorization}, when the batch compute times are stochastically decreasing and convex in the number of host workers, the balanced assignment achieves the minimum average job compute time, across all non-overlapping batch assignments.

In the following, we show that for all three service time distributions at workers the balanced assignment of non-overlapping batches minimizes the average job compute time. Several other assignments are proposed in the literature, cf. \cite{li2017near,tandon2016gradient}, which according to our results are not optimal.

\subsection{Exponential Distribution}
With exponential service time distribution at workers $T_{i,j}\sim\text{Exp}(\mu)$, the compute time of batch $i\in\{1,2,\dots,B\}$ is the minimum of $N_i$ i.i.d exponential RVs. Hence,
    \begin{equation}
        T_i\sim \text{Exp}\left(N_i\mu\right),\quad \forall i\in \{1,2,\dots,B\}.
    \label{equ:exp}
    \end{equation}
It can be verified that $T_i$ is stochastically decreasing and convex in $N_i$. Thus, the Theorem \ref{mainThm} follows.

\begin{theorem}
With i.i.d exponential service time distribution at workers $T_{i,j}\sim\text{Exp}(\mu)$, among all (non-overlapping) batch assignment policies, the balanced assignment achieves the minimum average job compute time.
\label{mainThm}
\end{theorem}

\subsection{Shifted-Exponential Distribution}
With shifted-exponential service time distribution at workers, the compute time of a batch consists of a deterministic part and a random part. The deterministic part can be thought of as the minimum required time to execute a batch. Thus, this part is batch dependent and varies with the size of a batch. On the other hand, the random part can be thought of as a worker dependent component. Note that, in this case the random part is i.i.d exponentially distributed across the workers. Accordingly, the compute time of a batch with $T_{i,j}\sim\text{Sexp}(\Delta,\mu)$, is a $\Delta$ plus the RV (\ref{equ:exp}). Thus, the Corollary \ref{cor:sExp1} follows from Theorem \ref{mainThm}.

\begin{corollary}\let\qed\relax
With shifted-exponential service time distribution at workers $T_{ij}\sim \text{SExp}(\Delta,\mu)$, among all (non-overlapping) batch assignment policies, the balanced assignment achieves the minimum average job compute time.
\label{cor:sExp1}
\end{corollary}

    \begin{figure*}[t]
        \centering
            \includegraphics[width=18cm, keepaspectratio]{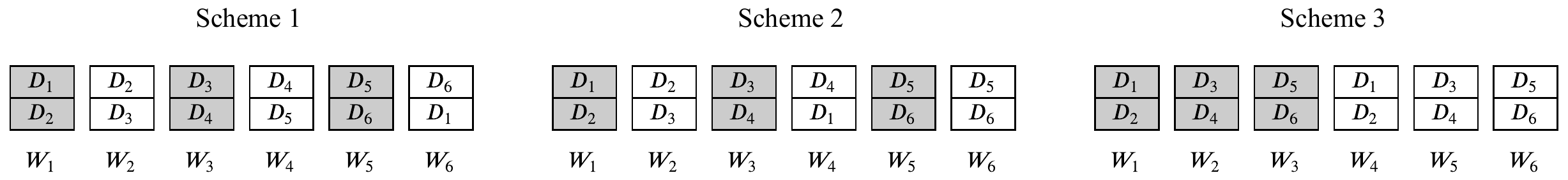}
            \caption{Three task-to-worker assignment schemes: 1) cyclic overlapping, 2) combination of cyclic overlapping and non-overlapping and 3) non-overlapping.}
        \label{fig:6workers}    
    \end{figure*}

\subsection{Pareto Distribution}
With pareto service time distribution at workers, the compute time of a batch consists of a deterministic part and random part. Similar to shifted-exponential case, the deterministic part can be associated with the batch size and the random part is worker dependent and may vary across the workers.
\begin{theorem}\let\qed\relax
Whit Pareto service time distribution at workers $T_{ij}\sim \text{Pareto}(\sigma,\alpha)$, among all (non-overlapping) batch assignment policies, the balanced assignment achieves the minimum average job compute time.
\label{thm:par0}
\end{theorem}

We have been able to show the desirable decreasing and convexity properties for our choice of distributions. One may study those behaviours for other service time distributions. Another possible extension of this work may study the optimal batch assignment, when the decreasing and convexity properties does not hold. Furthermore, the optimality of a batch assignment can be studied under different metrics, such as job compute time variability/predictability. However, due to space limitations we leave these problems open for future studies.

\section{Compute Time with Overlapping Batches} \label{section V}
We next show that, regardless of service time distribution at workers, the balanced assignment of non-overlapping batches achieves lower average job compute time compared to any overlapping assignments. Let's recall the original problem of assigning $N$ tasks redundantly among $N$ workers. Each worker is assigned with $N/B$  tasks, for a given parameter $B|N$. There are several schemes to group the $N$ tasks into $N$ batches of size $N/B$. Here, we focus on the schemes where the set of overlapping batches could be divided into subsets, such that each task appears exactly once in each subset. With these schemes, the number of workers assigned to a task is equal for all tasks. It's worth mentioning that, the non-overlapping batches can be thought of as an especial case of overlapping batches, where within each subset batches do not overlap.

As another especial case, consider the scheme where the set of tasks is divided into $N$ overlapping batches, each with size $N/B$, in a cyclic order, as follows. The first batch consists of tasks 1 through $N/B$, the second batch consists of tasks 2 through $N/B+1$, and so on. This is Scheme 1 in Fig.~\ref{fig:6workers}. Note that, the batches built with cyclic scheme and non-overlapping batches are the two end of a spectrum. With cyclic scheme the number of batches that share at least one task with any given batch is maximum, whereas with non-overlapping batches this number is minimum. Precisely, with cyclic scheme, each batch shares at least one common task with $2\left(N/B-1\right)$ other batches. With non-overlapping batches this number is $N/B-1$. With any other batching scheme that conforms to our assumption, this number is greater than $N/B-1$ and smaller than $2\left(N/B-1\right)$. As an example, Scheme 2 in Fig.~\ref{fig:6workers} consists of a cyclic part and a non-overlapping part. In what follows, we aim to compare the average job compute time with three different batching schemes, provided in Fig.~\ref{fig:6workers}. Although we provide the comparison for the especial case of $N=6$ and $B=3$, the extension of our method to general values of $N$ and $B$ is straightforward.  

Consider system~\ref{fig:sysModel1}, with $N=6$, $B=3$, and three different batching schemes, shown in Fig.~\ref{fig:6workers}. In each scheme, there are two subsets of batches which contain exactly one copy of every task. The subsets are shown by different colors. Let $X_i$ $\forall i\in\{1,2,\dots,6\}$ be the i.i.d RV of batch $i$ compute time. Let us assume, without loss of generality, that $W_1$ is the fastest worker, delivering its local result before the rest of the workers. Then the job compute time with Scheme 1 is
    \begin{equation}
        T^{(1)}=\min\left(\max\left(X_3,X_5\right),\max\left(X_2,X_4,X_6\right)\right).
    \label{4a}
    \end{equation}
With Scheme 2, the job compute time could be written as,
    \begin{align}
        T^{(2)}=\min( & \max(X_3,\min(X_5,X_6)),\nonumber\\
        &
        \max\left(\max\left(X_2,X_4\right),\min\left(X_5,X_6\right)\right).
    \label{4b}
    \end{align}
Comparing (\ref{4a}) and (\ref{4b}) gives $\mathbbm{E}[T^{(2)}]<\mathbbm{E}[T^{(1)}]$, since
    \begin{equation*}
            \mathbbm{E}[\max\left(X_3,\min\left(X_5,X_6\right)\right)]<\mathbbm{E}[\max\left(X_3,X_5\right)]
~~\text{and}
    \end{equation*}
      \begin{equation*}
        \begin{split}
            &\mathbbm{E}[\max\left(\min\left(X_5,X_6\right),\max\left(X_2,X_4\right)\right)]<\\
            &\hspace{5cm}\mathbbm{E}[\max\left(X_2,X_4,X_6\right)].
        \end{split}
    \end{equation*}
On the other hand, with Scheme 3,
    \begin{equation*}
        T^{(3)}=\max\left(\min\left(X_3,X_4\right),\min\left(X_5,X_6\right)\right).
    \label{4c}
    \end{equation*}
In order to be able to compare the job compute time with of Scheme 2 and Scheme 3, we rewrite (\ref{4b}) as follows:
    \begin{equation*}
        T^{(2)}=\max\left(\min\left(X_3,\max\left(X_2,X_4\right)\right),\min\left(X_5,X_6\right)\right).
    \end{equation*}
Since
    $
        \mathbbm{E}[\min\left(X_3,X_4\right)]<\mathbbm{E}[\min\left(X_3,\max\left(X_2,X_4\right)\right)].
    $
we have that $\mathbbm{E}[T^{(3)}]<\mathbbm{E}[T^{(2)}]$. Similarly, for the average job compute times of three batching schemes in Fig.~\ref{fig:6workers} we have
$
        \mathbbm{E}[T^{(3)}]<\mathbbm{E}[T^{(2)}]<\mathbbm{E}[T^{(1)}].
$
Note that, this result does not depend on service time distribution at workers. Thus, we can conclude that, the balanced assignment of non-overlapping batches achieves lower average of job compute time when compared to overlapping batch assignment. This is an important result because overlapping assignments have been often proposed in the literature, cf.\ \cite{tandon2016gradient} and \cite{amiri2019computation}. 

\section{Optimum Redundancy Level}\label{sectionVI}

In this section, we address the problem of finding the optimum level of redundancy for different service time distributions. We show that, the optimum redundancy level depends on this distribution, and that for a given distribution it is a function of the distribution's parameters. Furthermore, we show that, the the optimum redundancy level is not necessarily the same for average job compute time and predictability, and that there exists a trade-off between the two metrics when optimizing for the redundancy level.

In a system with $B$ batches and $N$ workers, each batch is replicated $N/B$ times, and we call $N/B$ the \textit{level of redundancy}. Recall the original problem of redundantly assigning tasks to workers in system~\ref{fig:sysModel1}. In one extreme case, all the tasks could be assigned to every worker. We refer to this case as \textit{full diversity}, as each task experiences maximum diversity in execution time by getting assigned to every worker. In the other extreme case, each task is assigned to only one worker. We refer to this case as \textit{full parallelism}, as no worker performs redundant executions. Between full diversity and full parallelism there is a spectrum of assignments, each with different redundancy levels. We refer to this spectrum as \textit{diversity-parallelism} spectrum. Note that, we make no specific assumption about the maximum allowable level of redundancy in the system. However, in practice, a system may fail if the load exceeds a certain threshold. The implicit assumption in the results of this section is that the system can perform normally under all levels of redundancy.

We are interested in finding the level of redundancy that is optimal according to two important performance metrics: 1) the average job compute time, and 2) the job compute time predictability. While the former has been defined in the preceding sections, we define the latter as the Coefficients of Variations (Cov) of job compute time. Among two jobs with two different coefficient of variation, the one with smaller CoV has more predictable compute time. With a given redundancy level, the tasks are assigned to workers such that the average job compute time is minimized, as studied in Section \ref{sectionIV}. An extension of our work may consider finding the optimum level of redundancy, where for a given redundancy level the task assignment is done with the objective of maximizing the job compute time predictability. 

The batch service time of at a worker depends the number of tasks in the batch (its size), which is determined by the redundancy level. Several models have been used to describe how the batch service time at a worker scales with the number of its tasks (see e.g. \cite{gardner2016better,zubeldia2020delay,peng2020diversityparallelism} and references therein). The model we use in this work is the most common in the coded computing literature (see e.g., \cite{peng2020diversityparallelism} and references therein), described as follows.

In our model, all tasks have identical computing size. We assume that the random  task service time is determined by the worker it is assigned to, and thus all tasks in a batch experience the same service time. Each worker takes a random time $\tau$ to execute a task, which is i.i.d.\ across the workers. Therefore, the service time of all tasks hosted by the same worker is the same realization of the random variable $\tau$,
and the random batch service time is $\frac{N}{B}\tau$. The service times of the replicas of the same task (at different workers) are i.i.d.\ i.e., different realizations of the random variable $\tau$. This model was called \textit{server dependent scaling} in \cite{peng2020diversityparallelism}.

The results derived based on this model could as well be used to optimize redundancy level in systems which follow the proposed model in \cite{gardner2016better}. According to \cite{gardner2016better}, the batch service time is the product of two factors: 1) an RV associated with the batch size, and 2) an RV associated with the worker's \textit{slowdown}. With equal size tasks, the batch size is constant across the workers, and the only randomness in its service time comes from the slowdown RV. Accordingly, with the slowdown $\tau$ and the special case of identical batch sizes of $N/B$, the batch service time at a worker is $\frac{N}{B}\tau$. Thus, in this special case, the batch service time has the same expression as in our model and our results could be used.

In the rest of the paper, we assume $N$ is even and greater than 4, unless otherwise is stated. The extension of the results to the odd numbers is straightforward. 

\subsection{Shifted-Exponential Distribution}
For shifted-exponential distribution of tasks' service time, the following theorem gives the optimum level of redundancy that minimizes the average job compute time.
\begin{theorem}\let\qed\relax
With shifted-exponential distribution of tasks' service time $\tau\sim SExp(\Delta,\mu)$, the optimum $B$, achieving the minimum average job compute time, is the solution of the following discrete unconstrained optimization problem,
    \begin{equation}
        \underset{B\in F_B}{min}\qquad \frac{N\Delta}{B}+\frac{1}{\mu}H_{(B,1)},
        \label{equ:sexpAvg}
    \end{equation}
where $F_B$ is the set of all feasible values for $B$.
\label{thm:sExp1}
\end{theorem}

Finding the optimum level of redundancy in (\ref{equ:sexpAvg}) requires a search in all the feasible values of $B$, the following insights can be easily seen. When $\Delta$ and $\mu$ are large, the term $N\Delta/B$ dominates the objective function. In that case, increasing $B$, or decreasing the redundancy level reduces the average job compute time. With larger $\Delta$ and $\mu$,  there is less uncertainty in the tasks' service time. Thus reducing the redundancy level reduces the average job compute time. On the other hand, for small values of $\Delta$ and $\mu$, the term $\frac{1}{\mu}H_{(B,1)}$ dominates the objective function, and thus higher redundancy is more beneficial. This is expected because with higher randomness in tasks' service time, the probability that a worker experiences a sever slowdown is higher, in which case, more redundancy is beneficial.

In order to reduce the complexity of the search algorithm in (\ref{equ:sexpAvg}), the following theorem establishes a connection between the optimum operating point in the diversity-parallelism spectrum and the parameters of tasks' service time distribution.

\begin{theorem}\let\qed\relax
With shifted-exponential distribution of tasks' service time $\tau\sim SExp(\Delta,\mu)$, the optimum operating point in the diversity-parallelism spectrum, achieving the minimum average job compute time, is
    \begin{itemize}
        \item[-] at full diversity when $\Delta\mu<1/N$,
        \item[-] at a middle point when $1/N\leq\Delta\mu\leq\sum_{k=N/2+1}^N1/k$,
        \item[-] at full parallelism when $\sum_{k=N/2+1}^N1/k<\Delta\mu$.
    \end{itemize}
\label{thm:sExp2}
\end{theorem}

When  $\Delta\mu>1/N,\sum_{k=N/2+1}^N1/k$, the following corollary further simplifies the optimization problem in (\ref{equ:sexpAvg}).

\begin{corollary}\let\qed\relax
With shifted-exponential distribution of tasks' service time $\tau\sim SExp(\Delta,\mu)$ and $1/N\leq\Delta\mu\leq\sum_{k=N/2+1}^N1/k$, the optimum operating point in the diversity-parallelism spectrum, achieving the minimum average job compute time, is the solution of the following discrete unconstrained optimization problem,
    \begin{equation}
        \underset{B\in F_B}{\text{min}}\quad|B-N\Delta\mu|,
    \label{equ:sexpSimple}
    \end{equation}
where $F_B$ is the set of all feasible values for $B$. Note that the complexity of solving (\ref{equ:sexpSimple}) is $O(\log |F_B|)$.
\label{cor:sExp2}
\end{corollary}
 
The average job compute time vs.\ $B$ for different values of $\mu$, when the tasks' service time follows shifted-exponential distribution, is plotted in Fig.~\ref{fig:sexpAvg}, for $N=100$ and $\Delta=0.05$. For this set of parameters, $1/N=0.01$ and $\sum_{k=N/2+1}^N1/k\approx0.69$. Hence, for $\mu<0.2$ full diversity should minimize the average job compute time. For $0.2\leq\mu\leq13.8$ the optimum point should be in the middle of the spectrum. Finally, for $\mu>13.8$ full parallelism should minimize the average job compute time. All these regions could be verified in Fig.~\ref{fig:sexpAvg}.

        \begin{figure}[t]
            \centering
            \includegraphics[width=\columnwidth]{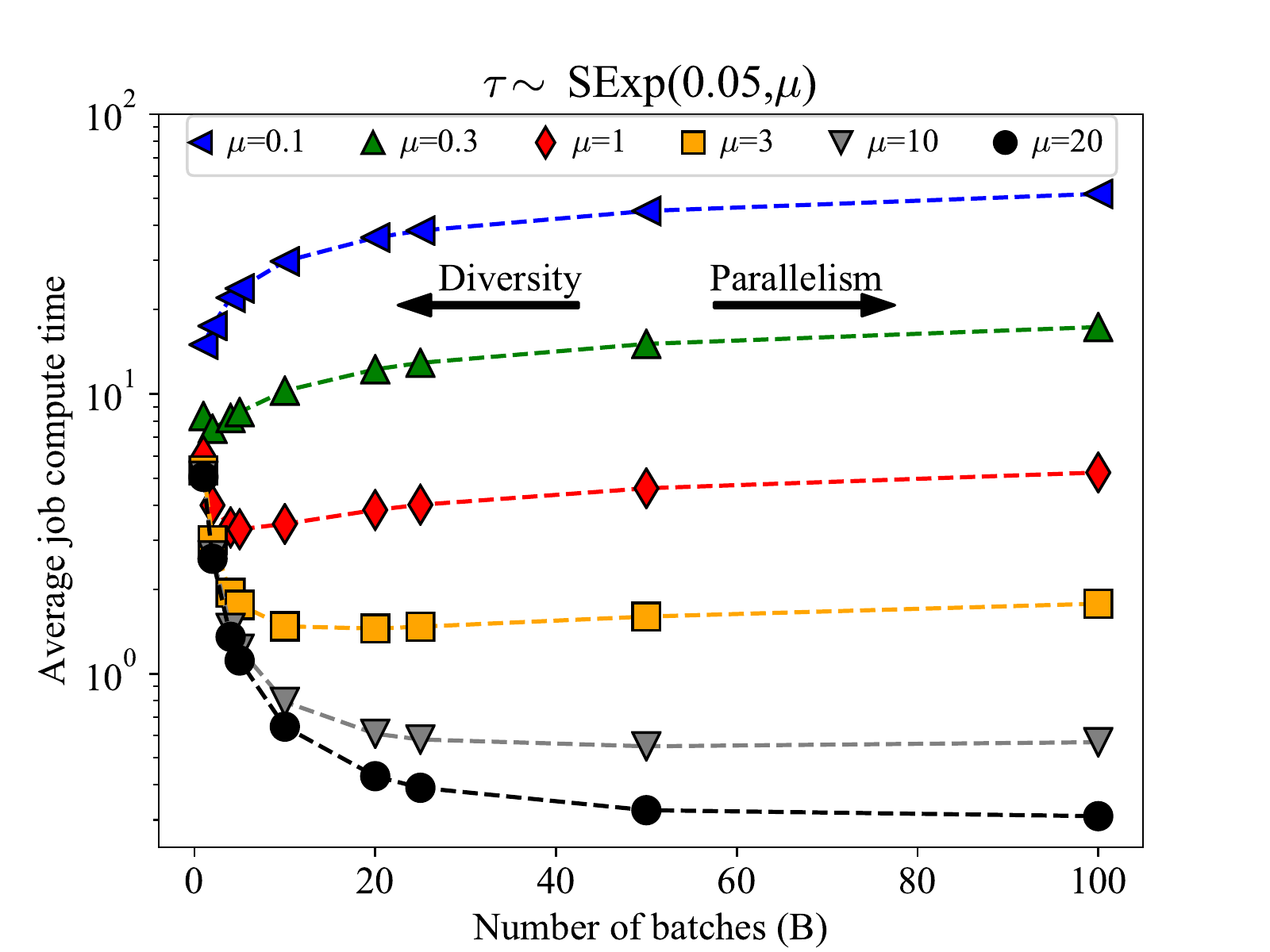}
            \caption{Average job compute time with $\tau\sim SExp(0.05,\mu)$, versus the number of batches, for different values of $\mu$. The minimum value of $\mathbbm{E}[T]$ moves toward the full parallelism as $\mu$ increases.}
            \label{fig:sexpAvg}
        \end{figure}

Next, we find the optimum level of redundancy that minimizes the coefficient of variations of job compute time. 

\begin{lemma}
With shifted-exponential distribution of tasks' service time $\tau\sim SExp(\Delta,\mu)$, the coefficient of variations of the job compute time is given by,
    \begin{equation}
        \text{CoV}[T]=\sqrt{H_{(B,2)}}\Bigl / \bigl[N\Delta\mu/{B}+H_{(B,1)}\bigr].
    \label{equ:sexpCov}
    \end{equation}{}
\label{lem:sexp}
\end{lemma}

\begin{theorem}\let\qed\relax
With shifted-exponential distribution of tasks' service time $\tau\sim SExp(\Delta,\mu)$, the optimum operating point in the diversity-parallelism spectrum, achieving the minimum coefficient of variations of job compute time, is
    \begin{itemize}
    \setlength\itemsep{3pt}
        \item[-] at full parallelism when $\Delta\mu<3/(\sqrt{5}-1)N$,
        \item[-] at either end of the spectrum when\\[2pt] $\frac{{3}}{(\sqrt{5}-1)N}\leq\Delta\mu\leq\frac{H_{(N,1)}\sqrt{H_{(N/2,2)}}-H_{(N/2,1)}\sqrt{(H_{(N,2)})}}{2\sqrt{H_{(N,2)}}-\sqrt{H_{(N/2,2)}}}$,
        \item[-] at full diversity when\\[2pt] $\Delta\mu>\frac{H_{(N,1)}\sqrt{H_{(N/2,2)}}-H_{(N/2,1)}\sqrt{(H_{(N,2)})}}{2\sqrt{H_{(N,2)}}-\sqrt{H_{(N/2,2)}}}$.
    \end{itemize}
\label{thm:sexpCov}
\end{theorem}

When tasks' service time follow i.i.d shifted-exponential distribution, the optimum operating point, minimizing the coefficient of variations of job compute time, is either full diversity or full parallelism. With higher randomness in the tasks' service time, i.e $\Delta\mu\in(-\infty,3/(\sqrt{5}-1)N)$, then assigning each worker with a single tasks minimizes the coefficient of variations of job compute time. On the other hand, with lower randomness in the tasks' service time, i.e. large $\Delta\mu$, full diversity $B=1$ is optimal. This is a sharp contrast with levels of redundancy that minimizes the average job compute time. Specifically, with high randomness, high redundancy level is required to minimize the average job compute time but low redundancy level is required to minimize the coefficient of variations of job compute time. In general, with shifted-exponential distribution of tasks' service time, one could not minimize both the expected value and the coefficient of variations of job compute time. This result sheds light on the impossibility of reducing average latency of compute jobs and maximizing the predictability of their service time in practical systems.

We simplify the problem of minimizing the coefficient of variations by the following corollary, which  gives the optimum level of redundancy for middle values of $\Delta\mu$ and large $N$.

    \begin{figure}[t]
        \centering
        \includegraphics[width=\columnwidth]{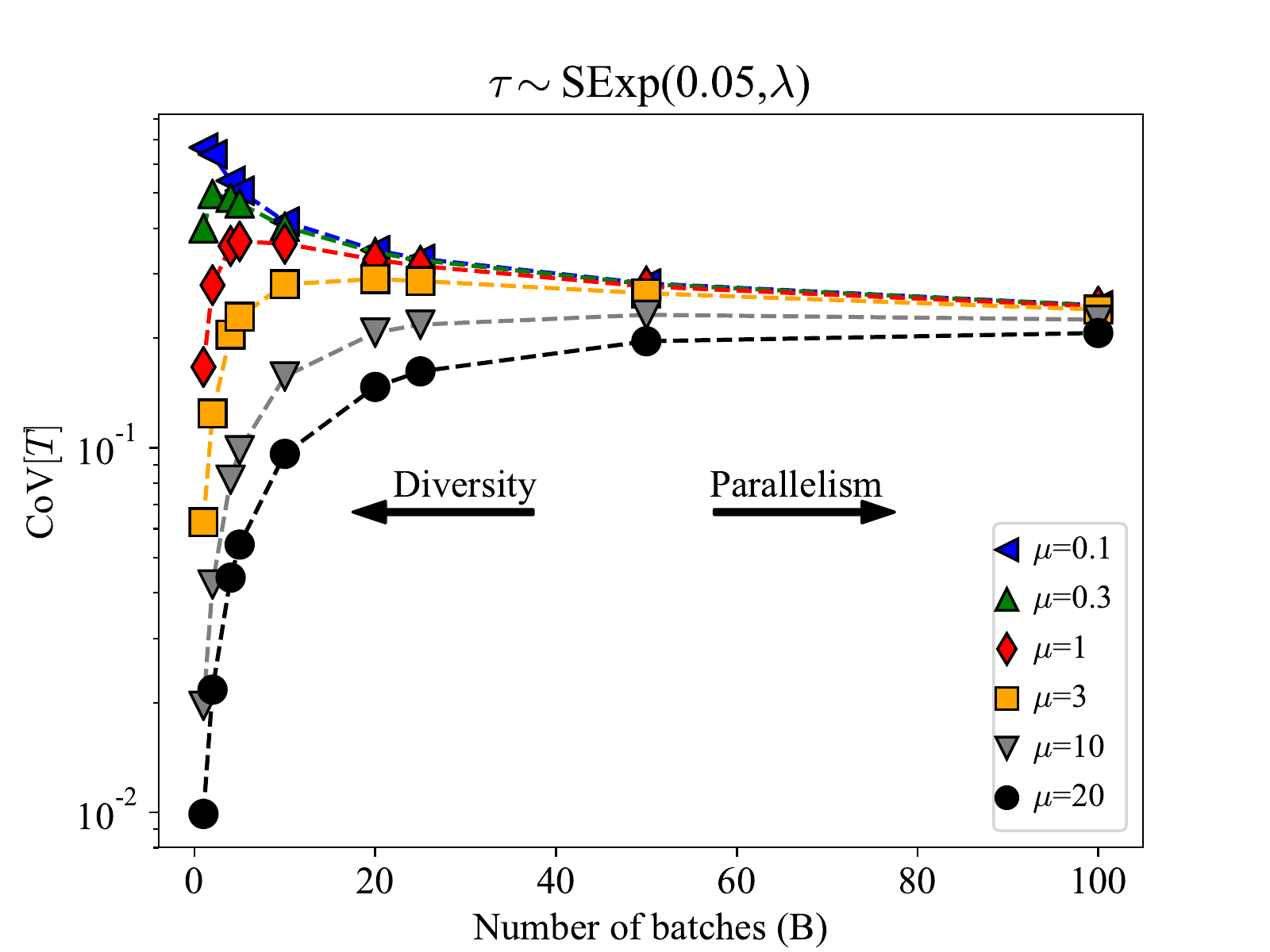}
        \caption{Coefficient of variations of job compute time with $\tau\sim SExp(0.05,\mu)$, versus the number of batches, for different values of $\mu$. The minimum value of $\text{CoV}[T]$ moves toward the full diversity as $\mu$ increases.}
        \label{fig:sexpCov}
    \end{figure}

\begin{corollary}\let\qed\relax
With shifted-exponential distribution of tasks' service time $\tau\sim SExp(\Delta,\mu)$, $N>11$ and\\[2pt] $\frac{{3}}{(\sqrt{5}-1)N}\leq\Delta\mu\leq\frac{H_{(N,1)}\sqrt{H_{(N/2,2)}}-H_{(N/2,1)}\sqrt{(H_{(N,2)})}}{2\sqrt{H_{(N,2)}}-\sqrt{H_{(N/2,2)}}}$,\\[1pt] the optimum operating point in the diversity-parallelism spectrum, achieving the minimum coefficient of variations of job compute time, is
    \begin{itemize}\setlength\itemsep{3pt}
        \item[-] at full parallelism when $\Delta\mu<H_{(N,1)}/(N\sqrt{H_{(N,2)}}-1)$, 
        \item[-] at full diversity when $H_{(N,1)}/(N\sqrt{H_{(N,2)}}-1)\leq\Delta\mu$.
    \end{itemize}
\label{cor:sexpCov}
\end{corollary}

The coefficient of variations of job compute time is plotted in Fig.~\ref{fig:sexpCov}, for $N=100$ and $\Delta=0.05$. For this set of parameters, $H_{(N,1)}/N(\sqrt{H_{n,2}}-1)\approx0.04$. Thus, for $\mu<0.04/\Delta\approx0.8$ full diversity and for $\mu>0.8$ full parallelism should be optimal. These regions can be verified in Fig.~\ref{fig:sexpCov}. Finally, we present the results for the limit case, when $N\rightarrow\infty$ in the following corollary.

\begin{corollary}
With $\tau\sim SExp(\Delta,\mu)$, the optimum operating point in the diversity-parallelism spectrum, achieving the minimum coefficient of variations of job compute time, occurs at full diversity, as $N\rightarrow\infty$.
\label{cor:sExp3}
\end{corollary}

With shifted-exponential distribution of tasks' service time we conclude that the expected value and the coefficient of variations of job compute time may not be optimized by 
the same redundancy level. For small and large values of $\Delta\mu$ product, the optimum points are at the opposite ends of the spectrum. In other words, the levels of redundancy that minimizes the average compute time results in the maximum coefficient of variations, and vice versa. Therefore, there is an inevitable trade-off between the average value and the coefficient of variations of job compute time, when tasks' service time follow shifted-exponential distribution. As a rule of thumb, when $\Delta\mu$ is small, the average job compute time is smaller at high diversity and the coefficient of variations of job compute time is smaller at high parallelism. Whereas, when $\Delta\mu$ is large, the average job compute time is smaller at high parallelism regime and the coefficient of variations is smaller at high diversity.

    \begin{figure}[t]
        \centering
        \includegraphics[width=\columnwidth]{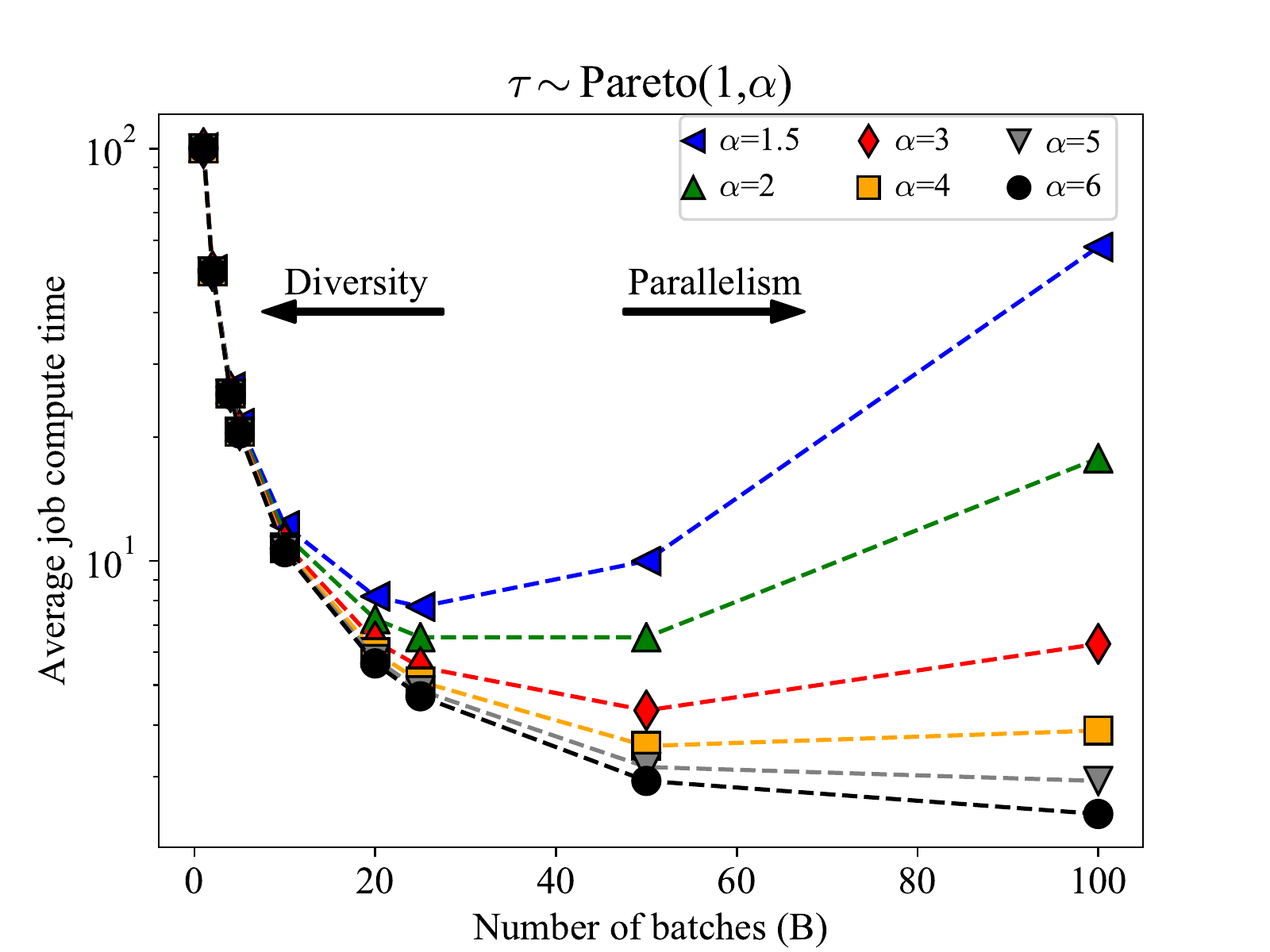}
        \caption{Average job compute time with $\tau\sim\text{Pareto}(1,\alpha)$, versus the number of batches, for different values of $\alpha$. The minimum value of $\mathbbm{E}[T]$ moves toward the full parallelism point as $\alpha$ increases.}
        \label{fig:paretoAvg}
    \end{figure}

\subsection{Pareto Distribution}
With Pareto distribution of tasks' service time, the following theorem gives the optimal level of redundancy that minimizes the average job compute time.

\begin{theorem}
With Pareto distribution of tasks' service time $\tau\sim Pareto(\sigma,\alpha)$, the optimum level of redundancy, achieving the minimum average job compute time, is the solution of the following discrete unconstrained optimization problem,
    \begin{equation}
        \underset{B\in F_B}{\text{min}}\qquad \frac{N\sigma}{B}\cdot\frac{\Gamma\left(B+1\right)\cdot\Gamma\left(1-B/N\alpha\right)}{\Gamma\left(B+1-B/N\alpha\right)},
        \label{equ:parAvg}
    \end{equation}
where $\Gamma(\cdot)$ is the Gamma function and $F_B$ is the set of all feasible values for $B$.
\label{thm:par1}
\end{theorem}

From \eqref{equ:parAvg}, when tasks' service time follow Pareto distribution, the average compute time grows linearly with the scale parameter $\sigma$. Nevertheless, its behaviour depends solely on the shape parameter $\alpha$. Therefore, the optimum level or redundancy is a function of $\alpha$, as follows.
    \begin{figure}[t]
        \centering
        \includegraphics[width=\columnwidth]{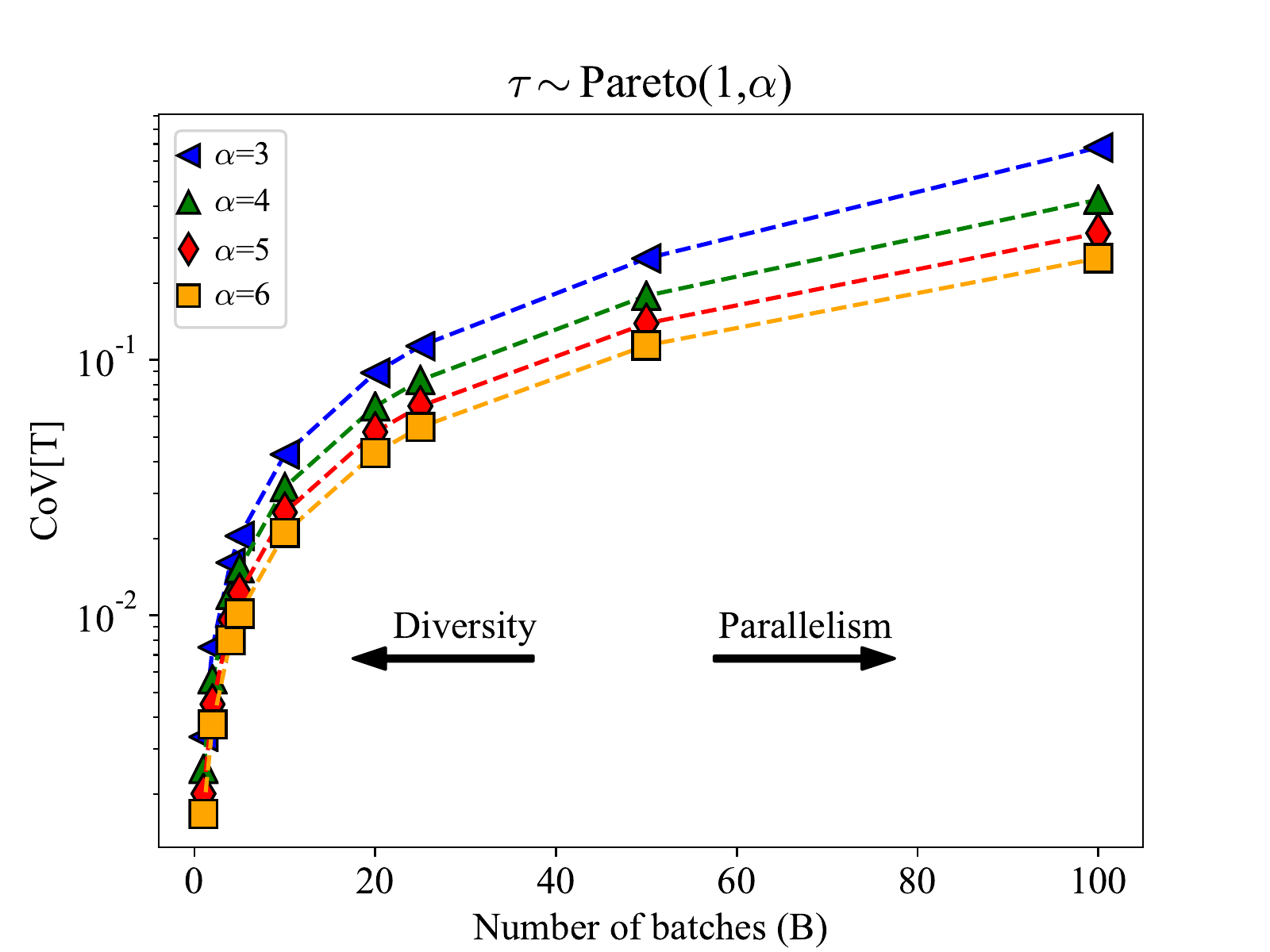}
        \caption{Coefficient of variations of job compute time with $\tau\sim\text{Pareto}(1,\alpha)$, versus the number of batches, for different values of $\alpha$. The minimum value of $\text{CoV}[T]$ is at the full diversity for all $\alpha>2$.}
        \label{fig:paretoCov}
    \end{figure}

\begin{theorem}
With Pareto distribution of tasks' service time $\tau\sim \text{Pareto}(\sigma,\alpha)$, the optimum operating point in the diversity-parallelism spectrum, achieving the minimum average job compute time, is
    \begin{itemize}
        \item[-] at a middle point when $1<\alpha<\alpha^*$, and
        \item[-] at full parallelism when $\alpha\geq\alpha^*$,
    \end{itemize}
where $\alpha^*$ is the solution of the following equation,
    \begin{equation}
        \frac{4\alpha^2+(\alpha-1)^2}{2\alpha(\alpha-1)}-\sqrt{\pi}N^{-1/2\alpha}2^{1+1/2\alpha}-0.58=0.
    \end{equation}
\label{thm:par2}
\end{theorem}
A Pareto RV with large shape parameter $\alpha$ has lighter tail and thus less randomness. Therefore, with large enough $\alpha$ full parallelism minimizes the average job compute time. With smaller values of $\alpha$, redundancy may be required to reduce the randomness in tasks' service time. Therefore, the optimal operating point should move towards the full diversity end of the spectrum. For $N=100$ and $\sigma=1$ the average job compute time is plotted in Fig.~\ref{fig:paretoAvg}. With this set of parameters, equation (\ref{equ:equ1}) has the solution $\alpha^*\approx4.7$. Thus, for $\alpha<4.7$ the optimum $B$ lies in a mid point of the diversity-parallelism spectrum. On the other hand, for $\alpha>4.7$ the optimum $B$ is at full parallelism, which can be verified by the plots in Fig.~\ref{fig:paretoAvg}.

\begin{lemma}
With Pareto distribution of tasks' service time $\tau\sim Pareto(\sigma,\alpha)$, the coefficient of variations of job compute time is given by,
    \begin{equation*}
        CoV(T)=\sqrt{\frac{\Gamma(B+1-B/N\alpha).\Gamma(1-2B/N\alpha)}{\Gamma(B+1-2B/N\alpha).\Gamma(1-B/N\alpha)}-1}.
    \end{equation*}
\label{lem:par1}
\end{lemma}

The coefficient of variations with Pareto distribution does not depend on the scale parameter $\sigma$ of the distribution. The following theorem gives the optimum level of redundancy that minimizes the coefficient of variation given by Lemma~\ref{lem:par1}.

\begin{theorem}
With Pareto distribution of tasks' service time $\tau\sim \text{Pareto}(\sigma,\alpha)$, the coefficient of variations of job compute time is minimized at full diversity.
\label{thm:par3}
\end{theorem}

Fig.~\ref{fig:paretoCov} shows the coefficient of variations of job compute time with Pareto distribution of tasks' service time. The optimum operating point is at full diversity, regardless of the value of $\alpha$. However, full diversity maximizes the average compute time for all values of $\alpha$, as it is shown in Fig.~\ref{fig:paretoAvg}. This result shows the trade-off between the expected value and the coefficient of variations of job compute time, with heavy-tail distribution of tasks' service time. 
For both exponential tail and heavy tail distributions of the workers' slow down, we have shown that there exist a trade-off between the average job compute time and the compute time predictability. In other words, minimizing the average latency of compute jobs and maximizing the predictability of this latency at the same time may not be possible. Therefore, in practical systems, this trade-off has to be considered in order to balance a reasonable balance between the average latency and predictability.

    \begin{figure}[t]
        \centering
            \includegraphics[width=\columnwidth, keepaspectratio]{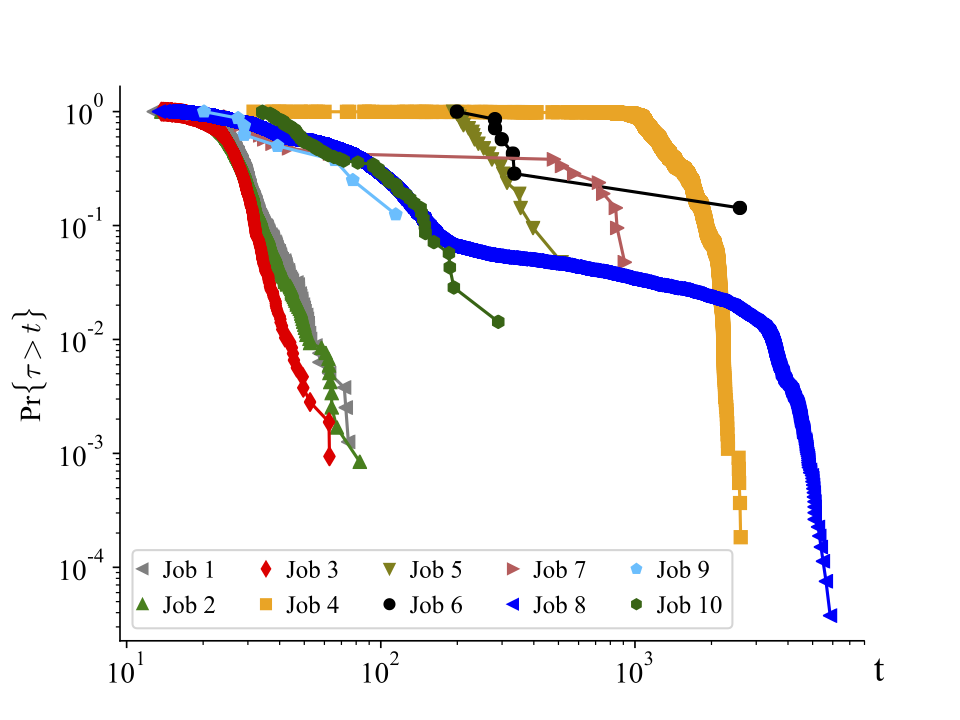}
            \caption{CCDF of the task compute time for 10 jobs. The runtime of the tasks are extracted from Google cluster traces dataset. }
        \label{fig:ccdf}
    \end{figure}

\section{Experiments}\label{sectionVII}
We next present results based on experiments on the dataset from Google cluster traces \cite{clusterdata:Wilkes2011} which provides the runtime information of the jobs in Google clusters. Jobs consists of multiple tasks, each executed by a worker. The recorded information for each task includes, among others, its scheduling and finish times. We recorded the service time of a task as the difference between its finish time and scheduling time.

We observed that the tasks' service time could follow both heavy-tail or exponential-tail behaviours, depending on the job. The task compute time CCDF is plotted in Fig \ref{fig:ccdf}. A linear decay in log-log scale means heavy-tail behaviour and exponential decay means exponential-tail. Accordingly, jobs 1 through 4 show exponential decay in tail probability, whereas jobs 5 through 10 have almost linear and thus heavy-tail decay.

To show the effect of redundancy level on the average compute time, we sampled tasks within a job. With the sampled tasks we formed batches and assigned each batch to a given number of workers, which is fixed across the batches. In Fig \ref{fig:exp_tail}, we plot the average job compute time, normalized by the average compute time with no-redundancy, versus the number of batches $B$.  Jobs 1--4, have exponential decay in the tail probability. For these jobs, the shift parameter for the shifted-exponential distribution is large (10 for jobs 1 -- 3 and 1000 for job 4). As predicted by our analysis, full parallelism minimizes the average job compute time.

In Fig. \ref{fig:heavy_tail} we plot the normalized average job compute time versus the number of batches, where task service times heavy-tailed. The minimum average job compute time occur somewhere between full parallelism and full diversity. This observation is inline with our analysis of Pareto distribution for tasks' service time. The optimum level of redundancy, however, depends on the job type. For instance, jobs 6, 8, 9 and 10 have their minimum average compute time at $B=20$ whereas for job 5 and 7 the optimum value is $B=50$ and $B=10$, respectively. This difference in the optimum level of redundancy is a result of the shape parameter of the distribution. For instance, job 7 decays faster than the other heavy-tail jobs and thus its task service time distribution has larger shape parameter. Therefore, it would require lower levels of redundancy. Note that, the heavy-tail behavior of task service times of some jobs in Fig. \ref{fig:ccdf} may not exactly fit into Pareto distribution. However, they are heavy-tail and our general results still apply.

    \begin{figure}[t]
        \centering
            \includegraphics[width=\columnwidth, keepaspectratio]{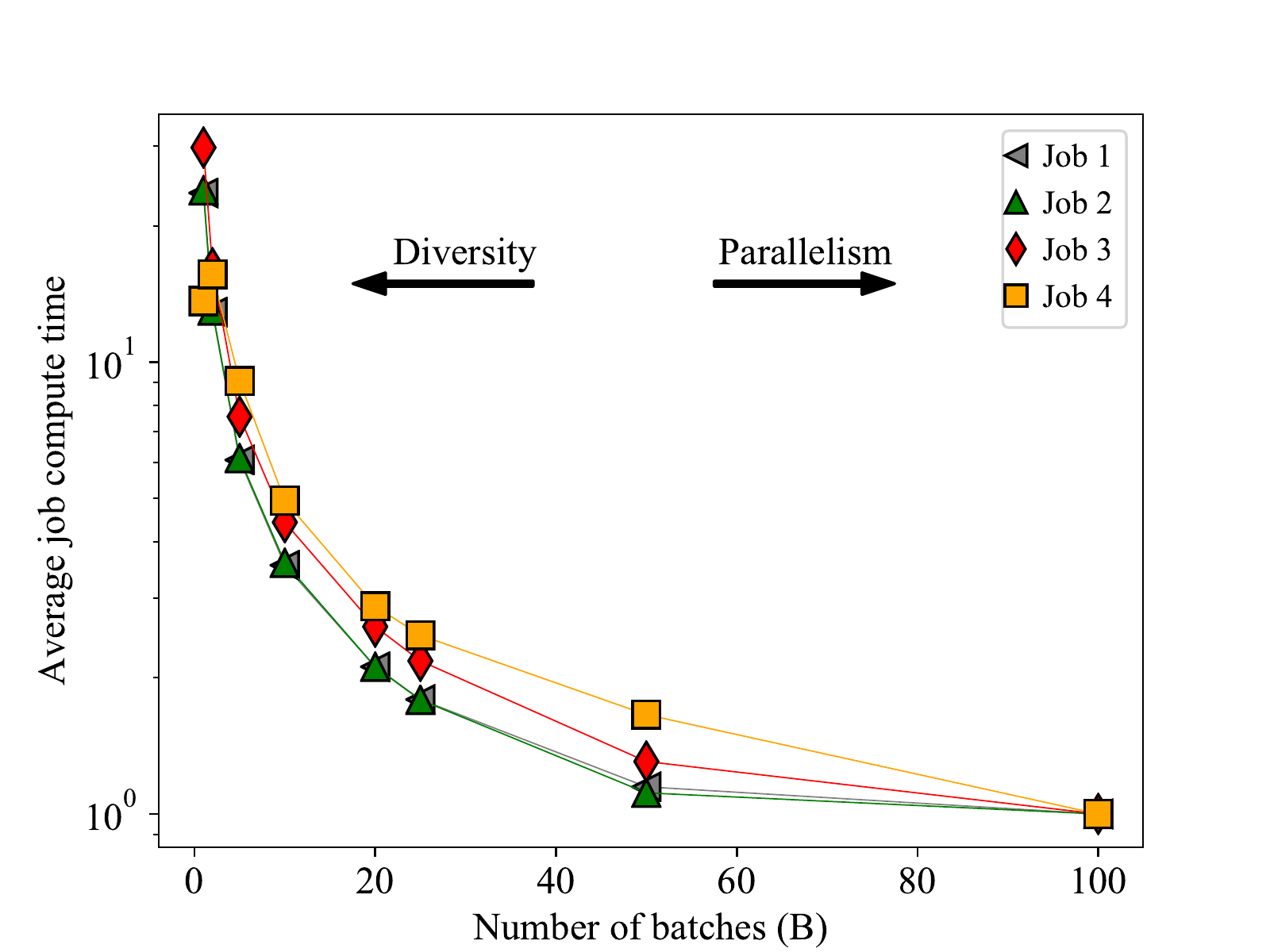}
            \caption{The effect of redundancy on the average job compute time, when the service time of tasks within the job have \textbf{exponential-tail} distribution.}
        \label{fig:exp_tail}
    \end{figure}

\section{Conclusion}\label{sectionVIII}
We studied the efficient task assignment problem in master-worker distributed computing system. For a given level of redundancy, we showed that if the batch compute times are i.i.d stochastically decreasing and convex in in the number of workers, a balanced assignment of non-overlapping batches achieves the minimum average job compute time. We then studied the optimum level of redundancy for minimizing average job compute time and maximizing compute time predictability. With both exponential-tail and heavy-tail distribution of workers' slow down, we showed that the redundancy level that minimizes the average job compute time is not necessarily the same as the redundancy level the maximizes the compute time predictability of jobs. Therefore, when optimizing for the optimum redundancy level, there exist an inevitable trade-off between the average and the predictability of job compute time. Finally, we evaluate the redundant task assignment with the data from Google cluster traces. We showed that a careful assignment of redundant tasks can reduce the average job compute time by an order of magnitude.

    \begin{figure}[t]
        \centering
            \includegraphics[width=\columnwidth, keepaspectratio]{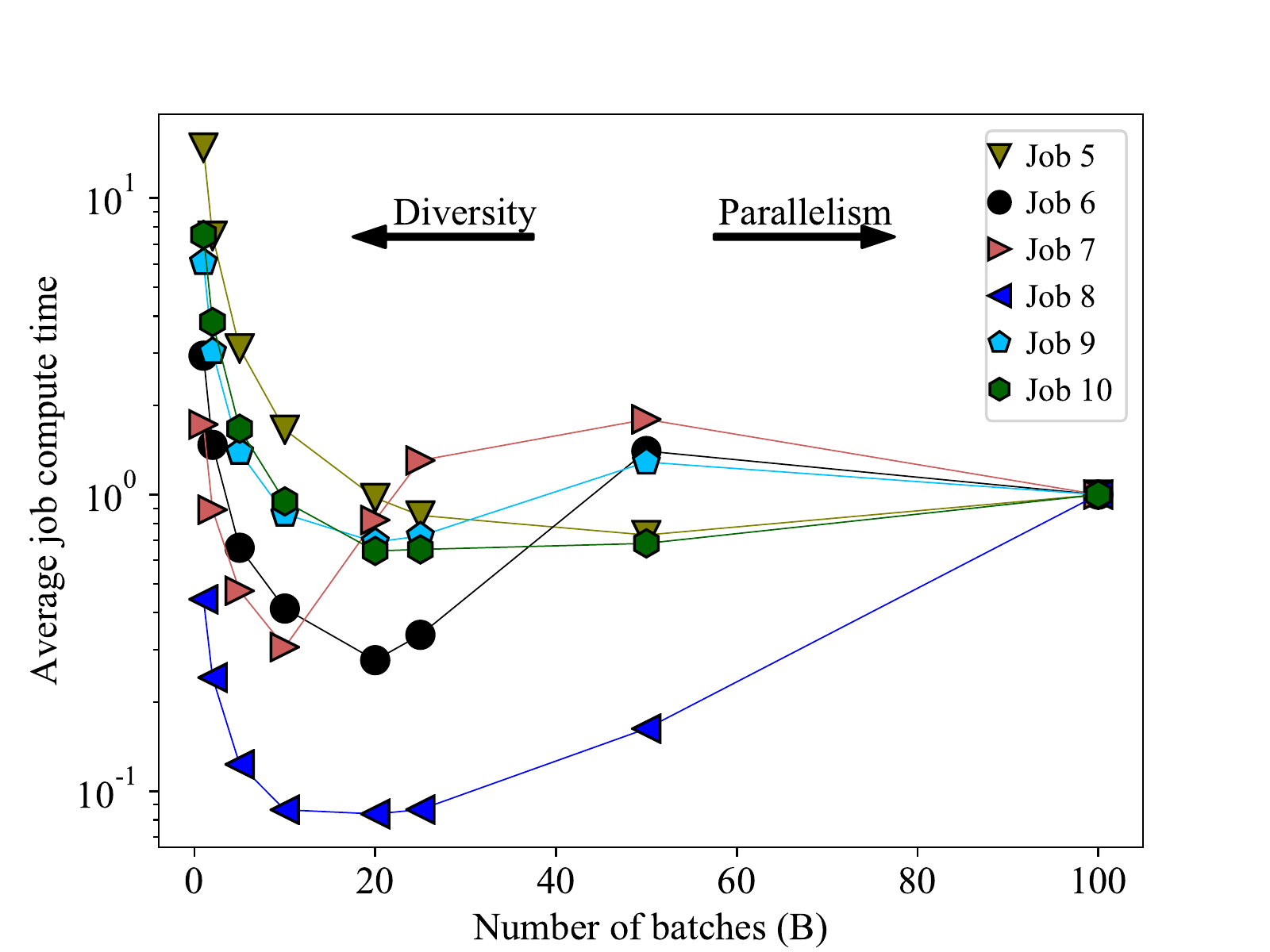}
            \caption{The effect of redundancy on the normalized average job compute time, when the task service time is \textbf{heavy tailed}.}
        \label{fig:heavy_tail}
    \end{figure}

\bibliographystyle{IEEEtran}
\bibliography{ref-3}

\newpage
\section*{Appendix}

\subsection{Random assignment of non-overlapping batches}\label{subsec:randAssign}

Let $n$ be the random number of workers required for covering all the batches in random assignment policy.

\begin{lemma}\label{lem:CC}
The probability of covering $B$ batches with $N$ workers with random batch-to-worker assignment is given by,
    \begin{equation}
        \textup{Pr}\{n\leq N\} = \frac{B!}{B^N}\stirlingii{N}{B},
    \label{covering}
    \end{equation}
where $\stirlingii{N}{B}$ is the Stirling number of second kind \cite{weisstein2002stirling}, given by,
    $
        \stirlingii{N}{B}=\frac{1}{B!}\sum_{k=0}^{B}(-1)^{B-k}\binom{B}{k}k^N.
    $
\end{lemma}

The probability in (\ref{covering}) of covering all the batches vs.\ the number of batches is plotted for four values of $N$ in Fig.~\ref{fig:stirling}. In order to cover all $B$ batches with high probability, the number of workers should be around one order of magnitude larger than the number of batches. Therefore, randomly assigning batches to workers is not a good practice, since the probability of being able to pick each batch at least once reduces very fast as $B$ increase. That could result in an inaccurate computations result, derived from only a subset of tasks.

\subsection{Proof of Lemma \ref{lem:CC}}
The probability of covering $B$ batches with exactly $N$ workers is given in \cite{myers2006some}, as
    \begin{equation}
        \textup{Pr}\{n= N\}=\frac{B!}{B^N}\stirlingii{N-1}{B-1}.
    \end{equation}
Therefore,
   \begin{equation*}
        \begin{split}
          \textup{Pr}\{n\leq N\}&=\sum_{n=B}^{N}\frac{B!}{B^n}\stirlingii{n-1}{B-1}\\
          &=B!\sum_{n=B-1}^{N-1}\frac{1}{B^{n+1}}\stirlingii{n}{B-1}\\
          &=\frac{B!}{B^N}\sum_{n=B-1}^{N-1}B^{N-n-1}\stirlingii{n}{B-1}\\
          &=\frac{B!}{B^N}\stirlingii{N}{B}.
        \end{split}
    \end{equation*}
    
    \begin{figure}[t]
        \centering
        \includegraphics[width=\columnwidth, keepaspectratio]{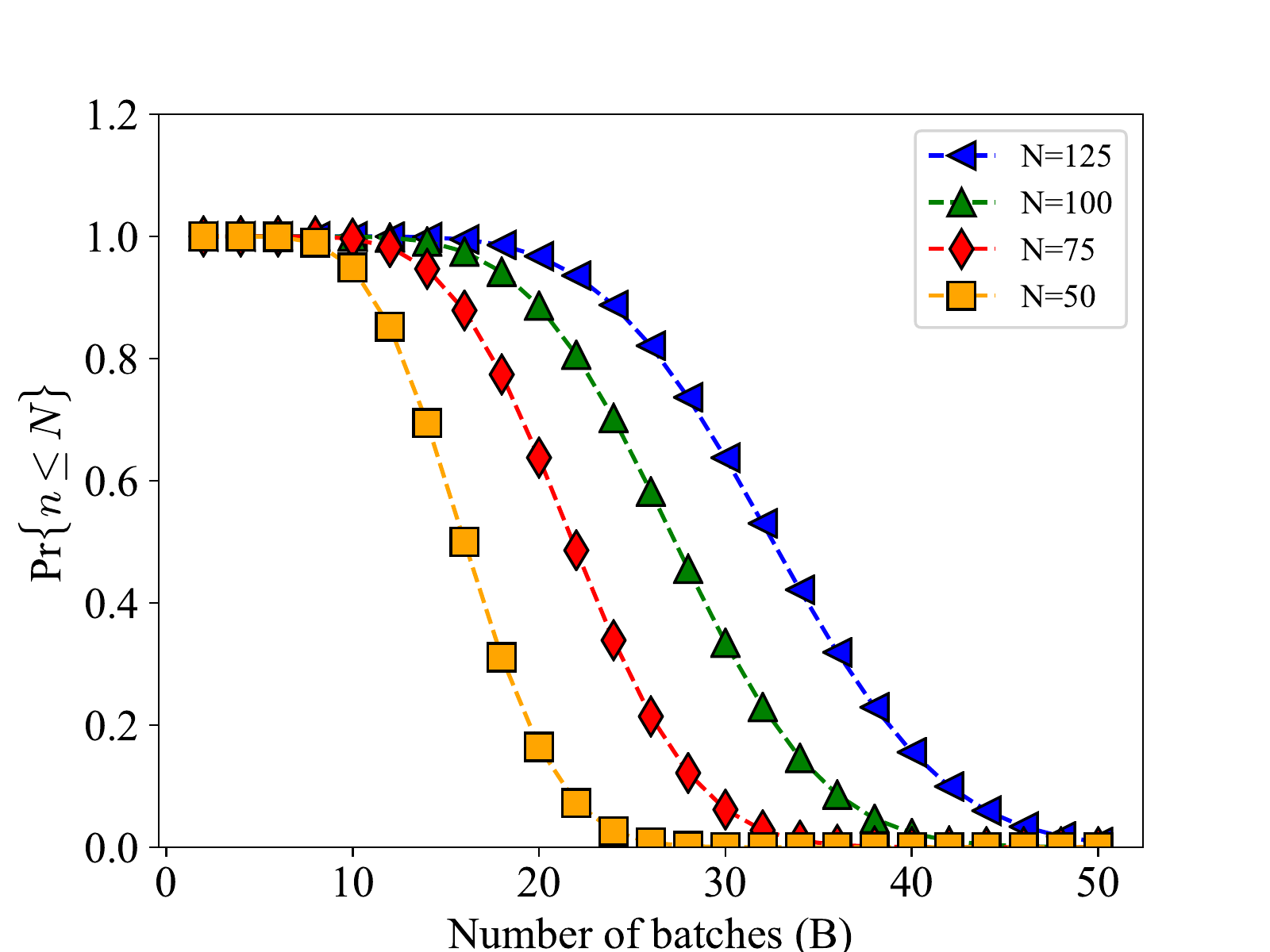}
        \caption{The probability of covering $B$ batches with $N$ workers, with random batch-to-worker assignment. The probability drops very fast as $B$ grows large. For $N=100$ only up to $B=10$ batches can be covered whp.}
   \label{fig:stirling}
    \end{figure}

\subsection{Proof of Theorem \ref{mainThm}}
According to Lemma \ref{lem}, since $T_i\sim \text{Exp}\left(N_i\mu\right)$. Given that $T_i$ RV is stochastically decreasing and convex in $N_i$, the minimum job compute time is achieved by balanced assignment of non-overlapping batches.

\subsection{Proof of Corollary \ref{cor:sExp1}}
CCDF of $T_{ij}\sim SExp(\Delta,\mu)$ is given by,
    $
        \Bar{F}_{T_{ij}}(t)=1-\mathbbm{1}(t\geq \Delta)\left[1-\textup{e}^{-\mu(t-\Delta)}\right]
    $
The minimum of $N_i$ i.i.d shifted-exponential RVs with rate $\mu$ and shift parameter $\Delta$ has also shifted-exponential distribution, with the following CCDF,
    \begin{equation}
        \Bar{F}_{T_{i}}(t)=1-\mathbbm{1}(t\geq \Delta)\left[1-\textup{e}^{-N_i\mu(t-\Delta)}\right]
    \end{equation}
Accordingly, $T$ could be written as,
    $
        \begin{aligned}
            T&=max\{T_1,T_2,\dots,T_B\}\\
             &=\Delta+max\{T^{\prime}_1,T^{\prime}_2,\dots,T^{\prime}_B\},
        \end{aligned}
    $
where $T^{\prime}_i\sim \text{Exp}(N_i\mu)$. Therefore, the expected job compute time could be written as,
    $
        \mathbbm{E}[T]=\Delta+\textit{E}[max\{T^{\prime}_1,T^{\prime}_2,\dots,T^{\prime}_B\}].
    $
Since $T'_i$s are independent exponential RVs, from Theorem \ref{mainThm}, the minimum expected job compute time is achieved by balanced assignment of non-overlapping batches.

\subsection{Proof of Theorem \ref{thm:par0}}
Suppose $T_{ij}\sim \text{Pareto}(\sigma,\alpha)$,
    $
        \Bar{F}_{T_{ij}}(t)=1-\mathbbm{1}(t\geq\sigma)\left[1-\left(\frac{t}{\sigma}\right)^{-\alpha}\right].
    $
It is easy to verify that $T_i$s follow a pareto distribution with,
    $
        \Bar{F}_{T_{i}}(t)=1-\mathbbm{1}(t\geq\sigma)\left[1-\left(\frac{t}{\sigma}\right)^{-N_i\alpha}\right].
    $
Therefore, the expected completion time could be written as,
    $
        \begin{aligned}
            \mathbbm{E}[T]&=max\{T_1,T_2,\dots,T_B\}\\
                &=\sigma+max\{T^{\prime\prime}_1,T^{\prime\prime}_2,\dots,T^{\prime\prime}_B\},
        \end{aligned}
    $
where $T^{\prime\prime}_i\sim \textup{Lomax}(N_i\alpha,\sigma)$. Since $T_i^{\prime\prime}$ is stochastically decreasing and convex in $N_i$, using the same arguments as in Theorem \ref{mainThm}, it can be concluded that the minimum expected job compute time, with Pareto service time distribution at workers, is achieved by the balanced assignment of non-overlapping batches.

\subsection{Proof of Theorem \ref{thm:sExp1}}
Suppose $\tau\sim SExp(\Delta,\mu)$. According to the size-dependent service time model, the service time of a worker with a batch of size $N/B$ will be the RV $T_{ij}=\frac{N}{B}\tau$. Thus, the CCDF of the batch compute time, that is available at $N/B$ workers, could be written as,
    
        \begin{align*}
            \Bar{F}_{T_{i}}(t)&=\textup{Pr}\{T_i>t\}\\
                        &=\textup{Pr}\{\min\left(T_{i1},T_{i2},\dots,T_{iN/B}\right)>t\}\\
                        &=\prod_{j=1}^{N/B}\textup{Pr}\{T_{ij}>t\}\\
                        &=\prod_{j=1}^{N/B}\textup{Pr}\{\tau>Bt/N\}\\
                        &=\left[\textup{Pr}\{\tau>Bt/N\}\right]^{N/B},\\
                        &=\left[\Bar{F}_\tau(Bt/N)\right]^{N/B}.
        \end{align*}

where,
    $
        \Bar{F}_{\tau}(Bt/N)=1-\mathbbm{1}(Bt/N\geq\Delta)\left[1-e^{-\mu\left(Bt/N-\Delta\right)}\right],
    $
and with further simplifications,
    $
        \Bar{F}_{\tau}(Bt/N)=1-\mathbbm{1}(t\geq N\Delta/B)\left[1-e^{-B\mu/N\left(t-N\Delta/B\right)}\right].
    $
Therefore,
    $
        \Bar{F}_{T_{i}}(t)=1-\mathbbm{1}(t\geq N\Delta/B)\left[1-\left[e^{-B\mu/N(t-N\Delta/B)}\right]^{N/B}\right],
    $
or equivalently,
    $
        \Bar{F}_{T_{i}}(t)=1-\mathbbm{1}(t\geq N\Delta/B)\left[1-e^{-\mu(t-N\Delta/B)}\right].
    $
In other words,
    $
        T_i\sim SExp(N\Delta/B,\mu).
    $
Note that, the CDF of $T_i$ is
    $
        F_{T_{i}}(t)=\mathbbm{1}(t\geq N\Delta/B)\left[1-e^{-\mu(t-N\Delta/B)}\right].
    $
Accordingly, the CDF of the job compute time, which is the maximum of $B$ i.i.d shifted-exponential RVs, could be written as,

        \begin{align*}
            F_T(t)&=\textup{Pr}\{T<t\},\\
                &=\textup{Pr}\{\max\left(T_1,T_2,\dots,T_B\right)<t\},\\
                &=\prod_{i=1}^{B}\textup{Pr}\{T_i<t\},\\
                &=\left[F_{T_i}(t)\right]^B.
        \end{align*}

Therefore,
    $
        F_{T}(t)=\mathbbm{1}(t\geq N\Delta/B)\left[1-e^{-\mu(t-N\Delta/B)}\right]^B.
    ${}
Thus, the average job compute time is,
    \begin{equation}
        \begin{aligned}
            \textit{E}[T]&=\int_{0}^{+\infty}\Bar{F}_T(t)dt,\\
                &=\int_{0}^{N\Delta/B}1dt+\int_{N\Delta/B}^{+\infty}\left[1-\left(1-e^{-\mu(t-N\Delta/B)}\right)^B\right]dt,\\
                &=\frac{N\Delta}{B}+\frac{1}{\mu}H_{(B,1)}.
        \end{aligned}
        \label{sexpst}
    \end{equation}
The function (\ref{sexpst}) is not monotonic in $B$. Therefore, to find the optimum $B$, one may solve the following discrete unconstrained optimization problem,
    $
        \underset{B\in F_B}{min} \qquad \frac{N\Delta}{B}+\frac{1}{\mu}H_{(B,1)},
    $
where $F_B$ is the set of all feasible values for $B$.

\subsection{Proof of Theorem \ref{thm:sExp2}}
Let's evaluate the average compute time (\ref{equ:sexpAvg}) at the two smallest operating points.
    \begin{align}
        & B=1:\quad\mathbbm{E}[T]=N\Delta+1/\mu\nonumber\\
        & B=2:\quad\mathbbm{E}[T]=N\Delta/2+3/2\mu\nonumber.
    \end{align}
For the average job compute time to be initially increasing in $B$ the function (\ref{equ:sexpAvg}) should have smaller value at $B=1$ than $B=2$, i.e. $\Delta\mu<1/N$. Evaluating this function at the two largest operating points,
    \begin{align}
        & B=N/2:\quad\mathbbm{E}[T]=2\Delta+H_{(N/2,1)}/\mu\nonumber\\
        & B=N:\quad\mathbbm{E}[T]=\Delta+H_{(N,1)}/\mu\nonumber,
    \end{align}
reveals that for it to ends increasing the product $\Delta\mu$ should satisfy,
    \begin{equation}
        \Delta\mu<H_{(N,1)}-H_{(N/2,1)},
    \end{equation}
or equivalently,
    \begin{equation}
        \Delta\mu<\sum_{k=N/2+1}^N1/k.
    \end{equation}
Now given that $1/N<\sum_{k=N/2+1}^N1/k$, if $\Delta\mu<1/N$ the function (\ref{equ:sexpAvg}) is monotonically increasing, which makes full diversity the optimum operating point. If $\Delta\mu>\sum_{k=N/2+1}^N1/k$ then the function (\ref{equ:sexpAvg}) is monotonically decreasing and full parallelism is the optimum operating point. Finally, if $1/N\leq\Delta\mu\leq\sum_{k=N/2+1}^N1/k$ then the function (\ref{equ:sexpAvg}) reaches a minimum point in neither end of the diversity-parallelism spectrum.

\subsection{Proof of Corollary \ref{cor:sExp2}}
Using the approximation of the harmonic number $H_B=\log B+\gamma$, the average compute time (\ref{equ:sexpAvg}) can be approximated by,
    \begin{equation}
        f(B)\coloneqq\frac{N\Delta}{B}+\frac{1}{\mu}\log B+\gamma/\mu,
    \label{equ:sexpAvgAprox}
    \end{equation}
where $\gamma$ is the Euler-Mascheroni constant. The minimum of (\ref{equ:sexpAvgAprox}) occurs at $B_{\text{min}}=N\Delta\mu$. Note that, since we assumed $B|N$, not every values of $B$ are feasible. In order to find the minimizing feasible value, we need to determine the value of $B$ which results in the minimum deviation form the value of (\ref{equ:sexpAvgAprox}) at $B_{min}$. We call this deviation $\delta\mathbbm{E}[T]$. The derivative of $f(B)$ is,
    $
        f'(B)=\frac{1}{\mu B}-\frac{N\Delta}{B^2}.
    $
Therefore,
    \begin{align}
        \delta\mathbb{E}[T]&=(B-B_{\text{min}})f'(B),\nonumber\\
        &=(B-N\Delta\mu)\left[\frac{1}{\mu B}-\frac{N\Delta}{B^2}\right],\\
        &=\frac{1}{\mu}\left(1-\frac{N\Delta\mu}{B}\right)^2.
    \end{align}
Accordingly a value of $B$ which minimizes $(1-N\Delta\mu/B)^2$, or equivalently $|B-N\Delta\mu|$, is the optimum operating point.

\subsection{Proof of Lemma \ref{lem:sexp}}
The proof is straightforward and is omitted.

\subsection{Proof of Theorem \ref{thm:sexpCov}}
Let's evaluate the coefficient of variations (\ref{equ:sexpCov}) at the two smallest operating points,
    \begin{align}
        &B=1:\quad \text{CoV}[T]=1/\Delta\mu N\nonumber\\
        &B=2:\quad \text{CoV}[T]=\sqrt{5}/(\Delta\mu N+3).\nonumber
    \end{align}
For the coefficients of variations to be initially increasing in $B$ the function (\ref{equ:sexpCov}) should have smaller value at $B=1$ than at $B=2$, i.e. $\Delta\mu<3/(\sqrt{5}-1)N$. Evaluating the same function at the two largest operating points,
    \begin{align}
        &B=N/2:\quad \text{CoV}[T]=\sqrt{H_{(N/2,2)}}/(2\Delta\mu+H_{(N/2,1)})\nonumber\\
        &B=N:\quad \text{Cov}[T]=\sqrt{H_{(N,2)}}/(\Delta\mu+H_{(N,1)}),\nonumber
    \end{align}
reveals that for it to ends increasing the product $\Delta\mu$ should satisfy,
    \begin{equation}
        \Delta\mu>\frac{H_{(N,1)}\sqrt{H_{(N/2,2)}}-H_{(N/2,1)}\sqrt{(H_{(N,2)})}}{2\sqrt{H_{(N,2)}}-\sqrt{H_{(N/2,2)}}}.
    \label{equ:sexpBound}
    \end{equation}
According to the limiting behaviour
    \begin{equation}
        \lim_{N\rightarrow\infty}\frac{H_{(N/2,2)}}{H_{(N,2)}}=1,
    \end{equation}
we can approximate the RHS of (\ref{equ:sexpBound}) as,
    \begin{equation}
    \begin{split}
        \lim_{N\rightarrow\infty}&\frac{H_{(N,1)}\sqrt{H_{(N/2,2)}}-H_{(N/2,1)}\sqrt{(H_{(N,2)})}}{2\sqrt{H_{(N,2)}}-\sqrt{H_{(N/2,2)}}}\\
        &\hspace{4cm}=H_{(N,1)}-H_{(N/2,1)}.
    \end{split}
    \end{equation}
Further, using the following series expansion at infinity
    \begin{equation}
        H_{(N,1)}-H_{(N/2,1)}=\log2-\frac{1}{2N}+\mathcal{O}(N^{-2})
    \end{equation}
it is easy to verify that $\forall N\in\{5,6,7,\dots\}$
    \begin{equation}
        3/(\sqrt{5}-1)N<H_{(N,1)}-H_{(N/2,1)}.
    \end{equation}
Therefore, for any integer $N>4$, the function (\ref{equ:sexpCov}) is monotonically decreasing when $\Delta\mu<3/(\sqrt{5}-1)N$ and thus full parallelism minimizes the coefficient of variations. It is monotonically increasing when $\Delta\mu>\frac{H_{(N,1)}\sqrt{H_{(N/2,2)}}-H_{(N/2,1)}\sqrt{(H_{(N,2)})}}{2\sqrt{H_{(N,2)}}-\sqrt{H_{(N/2,2)}}}$ and therefore full diversity minimized the coefficient of variations. Finally, when the product $\Delta\mu$ is between the two bounds, the function starts increasing and ends decreasing. In this case, the minimum coefficient of variations occur either at full diversity of at full parallelism.

\subsection{Proof of Corollary \ref{cor:sexpCov}}
From theorem \ref{thm:sexpCov}, with $3/(\sqrt{5}-1)N\leq\Delta\mu\leq\frac{H_{(N,1)}\sqrt{H_{(N/2,2)}}-H_{(N/2,1)}\sqrt{(H_{(N,2)})}}{2\sqrt{H_{(N,2)}}-\sqrt{H_{(N/2,2)}}}$, the coefficient of variations is minimum at one of the two ends of the diversity-parallelism spectrum. If the value of the function is smaller at $B=1$ than $B=N$ then full diversity is optimal. For that, the product $\Delta\mu$ should satisfy,
    \begin{equation}
        \Delta\mu<\frac{H_{(N,1)}}{N\sqrt{H_{(N,2)}}-1}.
    \label{equ:sexpBound1}
    \end{equation}
Otherwise full parallelism is optimum. The RHS of (\ref{equ:sexpBound1}) can be approximated by $H_{(N,1)}/N\sqrt{H_{(N,2)}}$. This approximation can be interpreted in terms of the derivatives of digamma function,
    \begin{equation}
        \frac{H_{(N,1)}}{\sqrt{H_{(N,2)}}}=\frac{\gamma\psi^{(0)}(N+1)}{\sqrt{\pi^2/6-\psi^{(1)}(N+1)}},
    \label{equ:digamma}
    \end{equation}
where $\psi^{(i)}(.)$ is the $i$th derivative of the digamma function. The RHS of (\ref{equ:digamma}) can be expanded around $\infty$ as,
    \begin{equation}
    \begin{split}
        \frac{\gamma\psi^{(0)}(N+1)}{\sqrt{\pi^2/6-\psi^{(1)}(N+1)}}&=\frac{\sqrt{6}\left(\gamma+\log N\right)}{\pi}\\
        &+\sqrt{\frac{3}{2}}\frac{6\gamma+6\log N+\pi^2}{\pi^3 N}\\
        &+\mathcal{O}\left(\frac{\log N}{N^2}\right).
    \end{split}
    \end{equation}
Therefore, for large values of $N$, the LHS of (\ref{equ:digamma}) can be well approximated by,
    \begin{equation}
        \frac{H_{(N,1)}}{\sqrt{H_{(N,2)}}}\approx\frac{\sqrt{6}\left(\gamma+\log N\right)}{\pi}+\sqrt{\frac{3}{2}}\frac{6\gamma+6\log N+\pi^2}{\pi^3 N}
    \label{equ:approx1}
    \end{equation}
It is easy to verify that (\ref{equ:approx1}) is increasing $\forall N\in\mathbbm{N}$. Moreover, $\forall N\in\{12,13,\dots\}$
    \begin{equation}
        \frac{3}{\sqrt{5}-1}<\frac{\sqrt{6}\left(\gamma+\log N\right)}{\pi}+\sqrt{\frac{3}{2}}\frac{6\gamma+6\log N+\pi^2}{\pi^3 N}.
    \end{equation}
Thus,
    \begin{equation}
        \frac{3}{(\sqrt{5}-1)N}<\frac{H_{(N,1)}}{N(\sqrt{H_{(N,2)}}-1)},\quad\forall N>11,
    \end{equation}
and for any value of $\Delta\mu$ that satisfies
    \begin{equation}
        \frac{3}{(\sqrt{5}-1)N}<\Delta\mu<\frac{H_{(N,1)}}{N(\sqrt{H_{(N,2)}}-1)}
    \end{equation}
$\text{CoV}[T]$ is minimized at full parallelism. Next we prove that $H_{(N,1)}/N\sqrt{H_{N,2}}$ is smaller than $H_{(N,1)}-H_{(N/2,1)}$ for $N>4$. It is easy to verify that,
    \begin{equation}
        \frac{N-1}{N/2+j}-\frac{1}{j}>0,\quad \forall N>4,j\in\mathbbm{N}.
    \end{equation}
Therefore,
    \begin{equation}
        (N-1)\sum_{k=N/2+1}^N1/k-\sum_{k=1}^{N/2}1/k>0,\quad\forall N>4,j\in\mathbbm{N}.
    \end{equation}
Thus we can write,
    \begin{equation}
        1+\frac{\sum_{k=N/2+1}^N1/k}{\sum_{k=1}^{N/2}1/k}>1+\frac{1}{N-1},\quad\forall N>4.
    \end{equation}
Accordingly,
    \begin{equation}
        H_{(N,1)}<N\left[H_{(N,1)}-H_{(N/2,1)}\right],\quad\forall N>4.
    \label{equ:ineq1}
    \end{equation}
Since
    $
        \sqrt{H_{(N,2)}}\geq1,\quad\forall N\in\mathbbm{N},
    $
we can update the bound in (\ref{equ:ineq1}) to,
    \begin{equation}
         H_{(N,1)}<N\left[H_{(N,1)}-H_{(N/2,1)}\right]\sqrt{H_{(N,2)}},\quad\forall N>4.
    \end{equation}
Consequently,
    \begin{equation}
        \frac{H_{(N,1)}}{N\sqrt{H_{(N,2)}}}<H_{(N,1)}-H_{(N/2,1)},
    \end{equation}
and thus for any value of $\Delta\mu$ that satisfies,
    \begin{equation}
        \frac{H_{(N,1)}}{N\sqrt{H_{(N,2)}}}<\Delta\mu<H_{(N,1)}-H_{(N/2,1)},\quad\forall N>4,
    \end{equation}
$\text{CoV}[T]$ is minimised at full diversity.

\subsection{Proof of Corollary \ref{cor:sExp3}}
Let's look at the asymptotic behaviour of the distance of $H_{(N,1)}/N(\sqrt{H_{n,2}}-1)$ to the two bounds in corollary \ref{cor:sexpCov},
    \begin{equation}
        \begin{split}
            \lim_{N\rightarrow\infty}&\frac{H_{(N,1)}-H_{(N/2,1)}-H_{(N,1)}/N(\sqrt{H_{n,2}}-1)}{H_{(N,1)}/N(\sqrt{H_{n,2}}-1)-3/(\sqrt{5}-1)N}\\
            &=\lim_{N\rightarrow\infty}\frac{\log 2-H_{(N,1)}/N(\sqrt{H_{n,2}}-1)}{H_{(N,1)}/N(\sqrt{H_{n,2}}-1)-3/(\sqrt{5}-1)N}\\
            &=\infty.
        \end{split}
    \end{equation}
Thus we have,
    \begin{equation}
        \Delta\mu\in\left(\frac{H_{(N,1)}}{N\sqrt{H_{(N,2)}}},H_{(N,1)}-H_{(N/2,1)}\right),\quad \text{as}\hspace{2mm}N\rightarrow\infty.
    \end{equation}
Therefore, from corollary \ref{cor:sexpCov}, the minimum $\text{CoV}[T]$ occurs at full diversity.

\subsection{Proof of Theorem \ref{thm:par1}}
Suppose $\tau\sim\textup{Pareto}(\sigma,\alpha)$. The CCDF of the service time of a batch with size $N/B$ is,
    $
    \begin{aligned}
        \textup{Pr}\{T_{ij}>t\}&=\textup{Pr}\{\tau>Bt/N\},\\
        &=1-\mathbbm{1}(t\geq N\sigma/B)\left[1-\left(\frac{Bt}{N\sigma}\right)^{-\alpha}\right].
    \end{aligned}
    ${}
And the CCDF of $T_i$ is,
    \begin{equation}
        \begin{aligned}
            \textup{Pr}\{T_{i}>t\}&=\left[\textup{Pr}\{T_{ij}>t\}\right]^{N/B},\\
            &=1-\mathbbm{1}(t\geq N\sigma/B)\left[1-\left(\frac{Bt}{N\sigma}\right)^{-N\alpha/B}\right].
        \end{aligned}{}
    \end{equation}{}
Hence, $T_i\sim\textup{Pareto}\left(N\sigma/B,N\alpha/B\right)$. From \cite{arnold2014pareto}, $\textit{E}[T]$, which is the maximum order statistics of $B$ RVs following $\textup{Pareto}(N\sigma/B,N\alpha/B)$, can be written as, 
    \begin{equation}
        \textit{E}[T]=\frac{N\sigma}{B}.\frac{\Gamma\left(B+1\right).\Gamma\left(1-B/N\alpha\right)}{\Gamma\left(B+1-B/N\alpha\right)},
    \label{paretoExp}
    \end{equation}
which completes the proof.

\subsection{Proof of Theorem \ref{thm:par2}}
For the objective function (\ref{equ:parAvg}) to be initially decreasing in $B$, it should be smaller at $B=2$ than at $B=1$.
    \begin{align}
        B=1:&\quad \mathbbm{E}[T]=N\sigma\frac{\Gamma(2)\Gamma(1-1/N\alpha)}{\Gamma(2-1/N\alpha)},\nonumber\\
        B=2:&\quad \mathbbm{E}[T]=\frac{N\sigma}{2}\frac{\Gamma(3)\Gamma(1-2/N\alpha)}{\Gamma(3-2/N\alpha)}\nonumber.
    \end{align}
In other words,
    \begin{equation}
        \frac{\Gamma(1-2/N\alpha)\Gamma(2-1/N\alpha)}{\Gamma(1-1/N\alpha)\Gamma(3-2/N\alpha)}>1.
        \label{equ:ineq2}
    \end{equation}
From the properties of Gamma function, it can be verified that,
    \begin{equation}
        \frac{\Gamma(1-2/N\alpha)\Gamma(2-1/N\alpha)}{\Gamma(1-1/N\alpha)\Gamma(3-2/N\alpha)}=\frac{1}{2(1-2/N\alpha)}.
    \end{equation}
By further simplification, (\ref{equ:ineq2}) yields to $\alpha>4/N$. Remember, we assumed that $N>4$. Thus $4/N<1$ and with $\alpha>1$ the objective function in (\ref{equ:parAvg}) is initially decreasing in $B$. Likewise, for the function (\ref{equ:parAvg}) to ends increasing it should be smaller at $B=N/2$ than at $B=N$.
    \begin{align}
        B=N/2:&\quad \mathbbm{E}[T]=2\sigma\frac{\Gamma(N/2+1)\Gamma(1-1/2\alpha)}{\Gamma(N/2+1-1/2\alpha)},\nonumber\\
        B=N:&\quad \mathbbm{E}[T]=\sigma\frac{\Gamma(N+1)\Gamma(1-1/\alpha)}{\Gamma(N+1-1/2\alpha)}\nonumber.
    \end{align}
In other words,
    \begin{equation}
        \frac{\Gamma(N+1)\Gamma(1-1/\alpha)}{\Gamma(N+1-1/2\alpha)}>\frac{2\Gamma(N/2+1)\Gamma(1-1/2\alpha)}{\Gamma(N/2+1-1/2\alpha)}.
    \label{equ:ineq3}
    \end{equation}
Substituting $\Gamma(N+1)=N\Gamma(N)$, the inequality (\ref{equ:ineq3}) can be rewritten as,
    \begin{equation}
        \frac{\Gamma(N)\Gamma(1-1/\alpha)}{\Gamma(N+1-1/2\alpha)}>\frac{\Gamma(N/2)\Gamma(1-1/2\alpha)}{\Gamma(N/2+1-1/2\alpha)}.
    \label{equ:ineq4}
    \end{equation}
According to the asymptotic behaviour of Gamma function
    \begin{equation}
        \lim_{N\rightarrow\infty}\frac{\Gamma(N+x)}{\Gamma(N)N^x}=1,
    \end{equation}
for large $N$ the inequality (\ref{equ:ineq4}) can be written as,
    \begin{equation}
        \frac{\Gamma(1-1/\alpha)}{N^{1-1/\alpha}}>\frac{\Gamma(1-1/2\alpha)}{(N/2)^{1-1/2\alpha}}.
    \end{equation}
Furthermore, it can be written that,
    \begin{equation}
        \frac{\Gamma(1-1/2\alpha)}{\Gamma(1-1/\alpha)}=\frac{2^{1/\alpha}\sqrt{\pi}}{\Gamma(1/2-1/2\alpha)}.
    \end{equation}
Accordingly, for large $N$ the inequality (\ref{equ:ineq4}) is written as,
    \begin{equation}
        \Gamma(1/2-1/2\alpha)>\sqrt{\pi}N^{-1/2\alpha}2^{1+1/2\alpha}
    \label{equ:ineq5}
    \end{equation}
Since $1/2-1/2\alpha<1/2$, we can further simplify (\ref{equ:ineq5}) by using the Laurent expansion of Gamma function around 0,
    \begin{equation}
        \begin{split}
            \Gamma(x)&=\frac{1}{x}-\gamma+\frac{\pi^2+6\gamma^2}{12}x+\mathcal{O}(x^2)\\
            &\approx\frac{1}{x}+x-0.58.
        \end{split}
    \end{equation}
Therefore, (\ref{equ:ineq3}) can be rewritten as,
    \begin{equation}
        \frac{4\alpha^2+(\alpha-1)^2}{2\alpha(\alpha-1)}-\sqrt{\pi}N^{-1/2\alpha}2^{1+1/2\alpha}-0.58>0.
    \end{equation}
It is easy to verify that in (\ref{equ:ineq5}) the RHS is a decreasing function of $\alpha$ and the LHS is an increasing function of $\alpha$, for $\alpha>1$. That means,
    \begin{equation}
        \Gamma(1/2-1/2\alpha)>\sqrt{\pi}N^{-1/2\alpha}2^{1+1/2\alpha},\hspace{2mm}\forall1<\alpha<\alpha^*,
    \end{equation}
    \begin{equation}
        \Gamma(1/2-1/2\alpha)<\sqrt{\pi}N^{-1/2\alpha}2^{1+1/2\alpha},\hspace{2mm}\forall\alpha\geq\alpha^*,
    \end{equation}
where $\alpha^*$ is the solution of
    \begin{equation}
        \frac{4\alpha^2+(\alpha-1)^2}{2\alpha(\alpha-1)}-\sqrt{\pi}N^{-1/2\alpha}2^{1+1/2\alpha}-0.58=0.
    \label{equ:equ1}
    \end{equation}

\subsection{Proof of Lemma \ref{lem:par1}}
From Theorem \ref{thm:par1} the distribution of the computing time of batch $i$ follows $\textup{Pareto}\left(N\sigma/B,N\alpha/B\right)$. The covariance of order statistics of $n$ $\textup{pareto}(\sigma',\alpha')$ RV is given in \cite{arnold2014pareto} as,
    $
    \begin{aligned}
        \text{Cov}[X_{k_1:n}&X_{k_2:n}] = \sigma'^2\frac{n!}{(n-k_2)!}\frac{\Gamma(n-k_1+1-2/\alpha')}{\Gamma(n+1-2/\alpha')}\\
        &\times\frac{\Gamma(n-k_2+1-1/\alpha')}{\Gamma(n-k_1+1-1/\alpha')} - \textit{E}[X_{k_1:n}]\textit{E}[X_{k_2:n}].
    \end{aligned}
    $
By setting $k_1=k_2=n$, the variance of the maximum order statistics is, 
    \begin{equation}
        \text{Var}[X_{n:n}] = \sigma'^2n!\frac{\Gamma(1-2/\alpha')}{\Gamma(n+1-2/\alpha')} - \textit{E}[X_{n:n}]^2.
    \end{equation}
Substituting (\ref{paretoExp}), $n=B$, $\sigma'=N\sigma/B$ and $\alpha'=N\alpha/B$,
    \begin{equation}
        \begin{aligned}
            \textit{Var}(T)= \left(\frac{N\sigma}{B}\right)^2&\frac{\Gamma\left(B+1\right)\Gamma\left(1-2B/N\alpha\right)}{\Gamma\left(B+1-2B/N\alpha\right)}\\
            &-\left(\frac{N\sigma}{B}.\frac{\Gamma\left(B+1\right).\Gamma\left(1-B/N\alpha\right)}{\Gamma\left(B+1-B/N\alpha\right)}\right)^2.
        \end{aligned}
    \label{equ:parVar}
    \end{equation}
Consequently,
    \begin{equation}
        \begin{split}
            \text{CoV}[T]&=\frac{\sqrt{\text{Var}[T]}}{\mathbbm{E}[T]}\\
            &=\sqrt{\frac{\Gamma(B+1-B/N\alpha)\Gamma(1-2B/N\alpha)}{\Gamma(B+1-2B/N\alpha)\Gamma(1-B/N\alpha)}-1}.
        \end{split}
    \end{equation}
    
\subsection{Proof of Theorem \ref{thm:par3}}
We define the two ratio
    \begin{align}
        &Q_1(B)=\frac{\Gamma(B+1-B/N\alpha)}{\Gamma(B+1-2B/N\alpha)},\label{equ:q1}\\
        &Q_2(B)=\frac{\Gamma(1-2B/N\alpha)}{\Gamma(1-B/N\alpha)}.
    \end{align}
We define the continuous version of $Q_1(B)$ as
    $$
        Q'_1(B)={\Gamma(1+B(1-1/N\alpha))}/{\Gamma(1+B(1-2/N\alpha))}.
    $$
    The derivative of $Q'_1(B)$ is,
    \begin{equation*}
        \begin{split}
            \frac{dQ'_1(B)}{dB}=\frac{Q'_1(B)}{N\alpha}\Big[&(N\alpha-1)\psi(1+B(1-1/N\alpha))\\
            &-(N\alpha-2)\psi(1+B(1-2/N\alpha))\Big]
        \end{split}
    \end{equation*}
which can be rewritten in the form of,
    \begin{equation}
        \begin{split}
            \frac{dQ'_1(B)}{dB}=\frac{Q'_1(B)}{N\alpha}\Big[&(N\alpha-1)\Big(\psi(1+B(1-1/N\alpha))\\
            &-\psi(1+B(1-2/N\alpha))\Big)\\
            &\quad+\psi(1+B(1-2/N\alpha))\Big].
        \end{split}
    \end{equation}
Here $\psi(.)$ is the digamma function, which is increasing in the positive real domain. Therefore,
    \begin{equation*}
        (N\alpha-1)\Big(\psi(1+B(1-1/N\alpha))-\psi(1+B(1-2/N\alpha))\Big)>0.
    \end{equation*}
For the set of parameters $B\geq1,\alpha>2,N>4$,
    \begin{equation}
        \psi(1+B(1-2/N\alpha))\geq\psi(1.75)>0.
    \end{equation}
Further, it is easy to verify that $Q'_1(B)>0$. Consequently, $Q'_1(B)$ is an increasing function, for our set of parameters. Given that $Q'_1(B)$ is smooth, we can argue that $Q_1(B)$ is also increasing. Likewise, we define the continuous counterpart of $Q_2(B)$ as
    \begin{equation*}
            Q'_2(B)=\frac{\Gamma(1-2B/N\alpha)}{\Gamma(1-B/N\alpha)}
            =\frac{2^{-2B/N\alpha}}{\sqrt{\pi}}\Gamma\left(1/2-B/N\alpha\right).
    \end{equation*}
We can write,
    \begin{equation}
        \begin{split}
            \frac{dQ'_2(B)}{dB}=-&\Big[\psi(1/2-B/N\alpha)+\log4\Big]\\
            &\times\frac{2^{-2B/N\alpha}\Gamma(1/2-B/N\alpha)}{N\alpha\sqrt{\pi}}.
        \end{split}
    \end{equation}
Furthermore,
\begin{align*}
     \psi(1/2-B/N\alpha)+\log4&<\psi(1/2)+\log4\\
            &\approx-0.57.
\end{align*}
           
Accordingly, it is true that $dQ'_2(B)/dB>0$ and thus $Q'_2(B)$ is increasing. Due to the smoothness of $Q'_2(B)$, we can argue that $Q_2(B)$ is also increasing. Finally, with both $Q_1(B)$ and $Q_2(B)$ being increasing, we can say that the function by Lemma~\ref{lem:par1} is increasing as well. Hence, the minimum coefficient of variations of completion time is achieved at minimum $B$, i.e. full diversity.

\end{document}